\documentclass[twocolumn]{article}

\usepackage{graphicx}
\usepackage{mathptmx}

\usepackage[a4paper,bindingoffset=0.2in,left=2.0cm,right=1.0in,top=1.0in,
bottom=1.0in,footskip=0.25in]{geometry}

\usepackage{lipsum}
\usepackage{comment}
\usepackage{color}
\usepackage{cite}

\usepackage{amsfonts}
\usepackage{amsmath}
\usepackage{amsthm}
\usepackage{amssymb}
\usepackage{url}
\usepackage{datetime}
\usepackage{lastpage}
\usepackage{graphicx}
\usepackage{lastpage}
\usepackage{bigints}
\usepackage{algorithm}
\usepackage{algpseudocode}
\usepackage{cmap}

\newtheorem{theorem}{Theorem}
\newtheorem{lemma}{Lemma}

\newcommand{\Y}{\mathbf{Y}}
\newcommand{\xvec}{\mathbf{x}}
\newcommand{\yvec}{\mathbf{y}}
\newcommand{\zvec}{\mathbf{z}}
\newcommand{\cvec}{\mathbf{c}}

\newcommand{\zerovec}{\mathbf{0}}
\newcommand{\epsvec}{\boldsymbol{\epsilon}}

\newcommand{\thetavec}{\boldsymbol{\theta}}

\newcommand{\Z}{\mathbb{Z}}

\DeclareMathOperator*{\argmax}{arg\,max\,}

\DeclareMathOperator*{\plim}{plim\,}

\algrenewcommand\textproc{}

\begin{document}

\title{High-dimensional structure learning of sparse vector autoregressive
    models using fractional marginal pseudo-likelihood
}

\author{
	Kimmo Suotsalo\footnote{Aalto University} \and
	Yingying Xu\footnote{RIKEN, Center for Advanced Intelligence Project
		AIP / iTHEMS} \and
	Jukka Corander\footnote{University of Oslo} \and
	Johan Pensar$^\ddagger$
}

\date{}
\maketitle

\begin{abstract}

Learning vector autoregressive models from multivariate time series is
conventionally approached through least squares or maximum likelihood
estimation. These methods typically assume a fully connected model which
provides no direct insight to the model structure and may lead to highly noisy
estimates of the parameters. Because of these limitations, there has been an
increasing interest towards methods that produce sparse estimates through
penalized regression.  However, such methods are computationally intensive and
may become prohibitively time-consuming when the number of variables in the
model increases. In this paper we adopt an approximate Bayesian approach to the
learning problem by combining fractional marginal likelihood and
pseudo-likelihood. We propose a novel method, PLVAR, that is both faster and
produces more accurate estimates than the state-of-the-art methods based on
penalized regression. We prove the consistency of the PLVAR estimator and
demonstrate the attractive performance of the method on both simulated and
real-world data.
\\ \\
\noindent\textbf{Keywords} Vector autoregression \textbullet \:
Pseudo-likelihood \textbullet \: Fractional marginal likelihood \textbullet \:
Gaussian graphical models \textbullet \: Multivariate time series

\end{abstract}

%===============================================================================

\section{Introduction}
\label{section:introduction}

Vector autoregressive (VAR) models
\cite{Lutkepohl-2005,Brockwell-2016,Neusser-2016} have become standard tools in
many fields of science and engineering. They are used in,
for example, economics \cite{Ang-2003,Ito-2008,Zang-2012}, psychology
\cite{Wild-2010,Bringmann-2013,Epskamp-2018}, and sustainable energy technology
\cite{Dowell-2015,Cavalcante-2017,Zhao-2018}. In neuroscience, VAR models have
been widely applied in brain connectivity analysis with data collected through
functional magnetic resonance imaging
\cite{Baccala-2001,Harrison-2003,Roebroeck-2005}, magnetoencephalography
\cite{Tsiaras-2011,Michalareas-2013,Fukushima-2015}, or electroencephalography
(EEG) \cite{Supp-2007,Gomez-2008,Chiang-2009}.

The conventional approach to learning a VAR model is to utilize either
multivariate least squares (LS) estimation or maximum likelihood (ML)
estimation \cite{Lutkepohl-2005}. Both of these methods produce dense model
structures where each variable interacts with all the other variables. Some of
these interactions may be weak but it is typically impossible to tell which
variables are directly related to each other, making the estimation results difficult
to interpret. Yet, the interpretation of the model dynamics is a goal of itself
in many applications, such as reconstructing genetic networks from gene
expression data \cite{Abegaz-2013}, or identifying functional connectivity
between brain regions \cite{Chiang-2009}.

Estimating a dense VAR model also suffers from another problem, namely, the large
number of model parameters may lead to noisy estimates and give rise to
unstable predictions \cite{Davis-2016}. In a dense VAR model, the number of
parameters grows linearly with respect to the lag length and quadratically with
respect to the number of variables. The exact number of parameters in such a
model equals $kd^2 + d(d+1)/2$ where $k$ denotes the lag length and $d$ is the
number of variables.

Due to the reasons mentioned above, methods that provide sparse estimates of
VAR model structures have received a considerable amount of attention in the
literature
\cite{Valdes-2005,Arnold-2007,Hsu-2008,Haufe-2010,Banbura-2010,Melnyk-2016}. A
particularly popular group of methods is based on penalized regression where a
penalty term is included to regularize the sparsity in the estimated model.
Different penalty measures have been investigated
\cite{Valdes-2005,Hsu-2008,Haufe-2010}, including the $L_1$ norm, the $L_2$
norm, hard-thresholding, smoothly clipped absolute deviation (SCAD)
\cite{Fan-2001}, and mixture penalty. Of these, using the $L_2$ norm is also
known as ridge regression, and using the $L_1$ norm as the least absolute
shrinkage and selection operator (LASSO) \cite{Tibshirani-1996}. While ridge
regression penalizes strong connections between the variables, it is incapable
of setting them to exactly zero. In contrast, LASSO and SCAD have the ability to force some of the parameter values to exactly zero. Hence,
these two methods are suitable for variable selection and, in the context of
VAR models, learning a sparse model structure while simultaneously estimating
its parameters.

Despite their usefulness in many applications, LASSO and SCAD also have certain
shortcomings with respect to VAR model estimation. LASSO is likely to give
inaccurate estimates of the model structure when the sample size is small
\cite{Arnold-2007,Melnyk-2016}. In particular, the structure estimates obtained
by LASSO tend to be too dense \cite{Haufe-2010}. While SCAD has been reported
to perform better \cite{Abegaz-2013}, it still leaves room for improvement in
many cases. In addition, both LASSO and SCAD involve one or more
hyperparameters whose value may have a significant effect on the results.
Suitable hyperparameter values may be found through cross-validation but this
is typically a very time-consuming task.

In order to overcome the challenges related to LASSO and SCAD, we propose a
novel method for learning VAR model structures. We refer to this method as
pseudo-likelihood vector autoregression (PLVAR) because it utilizes a
pseudo-likelihood based scoring function. The PLVAR method works in two steps:
it first learns the temporal structure of the VAR model and then proceeds to
estimate the contemporaneous structure. The PLVAR method is shown to yield
highly accurate results even with small sample sizes. Importantly, it is also
very fast compared with the other methods.

The contributions of this paper are summarised as follows. 1) We extend the
idea of learning the structure of Gaussian graphical models via
pseudo-likelihood and Bayes factors \cite{Leppa-aho-2017,Pensar-2017} from the
cross-sectional domain into the time-series domain of VAR models. The resulting
PLVAR method is able to infer the complete VAR model structure, including the
lag length, in a single run. 2) We show that the PLVAR method produces a
consistent structure estimator in the class of VAR models. 3) We present an
iterative maximum likelihood procedure for inferring the model parameters for a
given sparse VAR structure, providing a complete VAR model learning pipeline.
4) We demonstrate through experiments that the proposed PLVAR method is more
accurate and much faster than the current state-of-the-art methods.

The organization of this paper is as follows. The first section provides an
introduction to VAR models and their role in the recent literature. The second
section describes the PLVAR method, a consistency result for the method, and
finally a procedure for estimating the model parameters. The third section
presents the experiments and results, and the fourth section discusses the
findings and makes some concluding remarks. Appendices A and B provide a
search algorithm and a consistency proof, respectively, for the PLVAR estimator.

\subsection{VAR models}
\label{subsection:var_models}

In its simplest form, a VAR model can be written as
\begin{equation} \label{equation_VAR(k)_model}
    \yvec_t = \sum_{m = 1}^{k} \mathbf{A}_m \yvec_{t-m} + \epsvec_t
\end{equation}
where $\yvec_t$ is the vector of current observations, $\mathbf{A}_m$ are the
lag matrices, $\yvec_{t-m}$ are the past observations, and $\epsvec_t$ is a
random error. The observations are assumed to be made in $d$ dimensions. The
lag length $k$ determines the number of past observations influencing the
distribution of the current time step in the model. In principle, the random
error can be assumed to follow any distribution, but it is often modeled using a
multivariate normal distribution with no dependencies between the time steps.
In this paper we always assume the error terms to be independent and
identically distributed (i.i.d.) and that $\epsvec_t \sim
N(\zerovec,\boldsymbol{\Sigma})$.

It is also possible to express a VAR model as
\begin{equation} \label{equation_VAR_model_with_trend_and_constant}
    \yvec_t = \sum_{m = 1}^{k} \mathbf{A}_m \yvec_{t-m} + \mathbf{b} t + \cvec +
    \epsvec_t
\end{equation}
where $\mathbf{b}t$ accounts for a linear time trend and $\cvec$ is a constant.
However, for time series that are (at least locally) stationary, one can always
start by fitting a linear model to the data. This gives estimates of
$\mathbf{b}$ and $\cvec$, denoted as $\hat{\mathbf{b}}$ and $\hat{\cvec}$,
respectively. The value $\hat{\mathbf{b}} t + \hat{\cvec}$ can then be
subtracted from each $\yvec_t$, thereby reducing the model back to
(\ref{equation_VAR(k)_model}).

A distinctive characteristic of VAR models is their linearity. The fact
that the models are linear brings both advantages and drawbacks. One of the
biggest advantages is the simplicity of the models that allows closed-form
analytical derivations, especially when the error terms are i.i.d. and
Gaussian. The simplicity of the models also makes them easy to interpret: it is
straightforward to see which past observations have effect on the current
observations. A clear drawback of VAR models is that linear mappings may fail
to fully capture the complex relationships between the variables and the
observations. In practice, however, VAR models have shown to yield useful
results in many cases: some illustrative examples can be found in
\cite{Harrison-2003,Gomez-2008,Wild-2010,Michalareas-2013}.

\paragraph{Relationship with state space models}

VAR models are also related to linear state space models. A VAR($k$) model of
the form (\ref{equation_VAR(k)_model}) can always be written
\cite{Lutkepohl-2005} as an equivalent VAR(1) model
\begin{equation} \label{equation_VAR(1)_model}
    \Y_t = \mathbf{D} \Y_{t-1} + \mathbf{E}_t
\end{equation}
where
\begin{equation}
    \mathbf{Y}_t =
    \begin{bmatrix}
        \yvec_t \\ \yvec_{t-1} \\ \vdots \\ \yvec_{t-k+1}
    \end{bmatrix}, \:\:\:
    \mathbf{D} =
    \begin{bmatrix}
        \mathbf{A}_1 & \mathbf{A}_2 & \cdots & \mathbf{A}_{k-1} & \mathbf{A}_k
        \\
        \mathbf{I}   & \mathbf{0}   & \cdots & \mathbf{0} & \mathbf{0} \\
        \mathbf{0}   & \mathbf{I}   &        & \mathbf{0} & \mathbf{0} \\
        \vdots       &              & \ddots &            & \vdots \\
        \mathbf{0}   & \mathbf{0}   & \cdots & \mathbf{I} & \mathbf{0} \\
    \end{bmatrix}, \:\:\:
    \mathbf{E}_t =
    \begin{bmatrix}
        \epsvec_t \\ \mathbf{0} \\ \vdots \\ \mathbf{0}
    \end{bmatrix}.
\end{equation}
As \cite{Lutkepohl-2005} points out, (\ref{equation_VAR(1)_model}) can be
seen as the transition equation of the state space model
\begin{align}
    \xvec_t &= \mathbf{G}_t \xvec_{t-1} + \mathbf{u}_t \\
    \yvec_t &= \mathbf{H}_t \xvec_t + \mathbf{v}_t
\end{align}
where $\mathbf{H}_t = \mathbf{I}$ and $\mathbf{v}_t = \zerovec$. This
can be interpreted as the system state being equal to the observed values, or
equivalently, as a case where one does not include a hidden state in the model.
While a hidden state is a natural component in certain applications, its
estimation requires computationally expensive methods such as Kalman filters.
Also, in some cases, the transition equation may not be known or is not backed
up with a solid theory. In such cases VAR modeling allows the estimation and
the analysis of the relationships between the current and the past
observations. Although the VAR modeling approach concentrates on the
observations rather than the underlying system that generates them, it can
still provide highly useful results such as implications of Granger causality
or accurate predictions of future observations.

\subsection{Graphical VAR modeling}

Similar to \cite{Dahlhaus-2003}, this paper utilizes a time series chain (TSC)
graph representation for modeling the structure of the VAR process. Let $Y = \{
Y_t^i:t \in \Z, i \in V \}$, where $V = \{ 1,\ldots,d \}$, represent a
stationary $d$-dimensional process. For any $S \subseteq V$, let $Y^S$ denote
the subprocess given by the indices in $S$ and let $\overline Y_t^S = \{ Y_u^S:
u < t \}$ denote the history of the subprocess at time $t$. In the TSC graph
representation by \cite{Dahlhaus-2003}, the variable $Y_t^i$ at a specific time
step is represented by a separate vertex in the graph. More specifically, the
TSC graph is denoted by $G_{TS} = (V_{TS}, E_{TS})$, where $V_{TS} = V \times
\Z$ is a time-step-specific node set, and $E_{TS}$ is a set of directed and
undirected edges such that
\begin{equation} \label{eqn:temporal_markov_blanket}
\begin{aligned}
&(a, t-u) \to (b,t) \notin E_{TS} \\
&\Leftrightarrow \\
&u \leq 0 \ \ \text{ or } \ \ Y^a_{t-u} \perp Y^b_t \:|\: \overline{Y}_t \setminus \{Y^a_{t-u}\}
\end{aligned}
\end{equation}
and
\begin{equation} \label{eqn:contemp_markov_blanket}
\begin{aligned}
&(a, t-u) \leftrightarrow (b,t) \notin E_{TS} \\
&\Leftrightarrow \\
&u \neq 0 \ \ \text{ or } \ \ Y^a_t \perp Y^b_t \:|\: \overline{Y}_t \cup \{Y^{V\setminus\{a,b\}}_t\}.
\end{aligned}
\end{equation}
As noted in \cite{Dahlhaus-2003}, conditions \eqref{eqn:temporal_markov_blanket} and \eqref{eqn:contemp_markov_blanket}
imply that the process adheres to the pairwise AMP Markov property
 for $G_{TS}$ \cite{Andersson-2001}.

In terms of a VAR process of form \eqref{equation_VAR(k)_model} with $\epsvec_t \sim N(\zerovec,\boldsymbol{\Sigma})$, the conditional independence relations in \eqref{eqn:temporal_markov_blanket} and \eqref{eqn:contemp_markov_blanket} correspond to simple restrictions on the lag matrices $\mathbf{A}_m$, $m=1,\ldots,k,$ and the precision matrix $\boldsymbol{\Omega} = \boldsymbol{\Sigma}^{-1}$ \cite{Dahlhaus-2003}. More specifically, we have that
\begin{equation*}
(a, t-u) \to (b,t) \in E_{TS} \Leftrightarrow  u \in \{1,\ldots,k \} \text{ and
} \mathbf{A}_u(b,a) \neq 0
\end{equation*}
and
\begin{equation*}
(a, t) \leftrightarrow (b,t) \in E_{TS} \Leftrightarrow \boldsymbol{\Omega}(a,b)
\neq 0.
\end{equation*}
Thus, finding the directed and undirected edges in $G_{TS}$ is
equivalent to finding the non-zero elements of the lag matrices $\mathbf{A}_m$
and the non-zero off-diagonal elements of the precision matrix
$\boldsymbol{\Omega}$, respectively.

As a concrete example of the above relationships, consider the TSC-graph in Figure \ref{fig:mixed_graph}(a) which represents the dependence structure of the following sparse VAR(2) process:

\begin{align*}
\mathbf{A}_1 &=
    \begin{bmatrix}
        0.3 & 0 & 0 & 0\\
        -0.2 & 0.2 & 0 & 0\\
        0 & 0 & -0.3 & 0\\
        0 & 0 & 0.2 & -0.2\\
    \end{bmatrix}, \quad
\mathbf{A}_2 =
    \begin{bmatrix}
        0 & 0.1 & 0 & 0\\
        0 & 0 & 0 & 0\\
        0 & 0 & 0 & -0.1\\
        0 & 0 & 0 & 0\\
    \end{bmatrix},
\\
\boldsymbol{\Omega} &=
    \begin{bmatrix}
        1 & 0 & 0.2 & 0\\
        0 & 1 & 0 & 0\\
        0.2 & 0 & 1 & 0.2 \\
        0 & 0 & 0.2 & 1\\
    \end{bmatrix}.
\end{align*}

\begin{figure}
    \begin{center}
        \includegraphics[width=0.48\textwidth, trim={75 190 210 165}, clip]
        {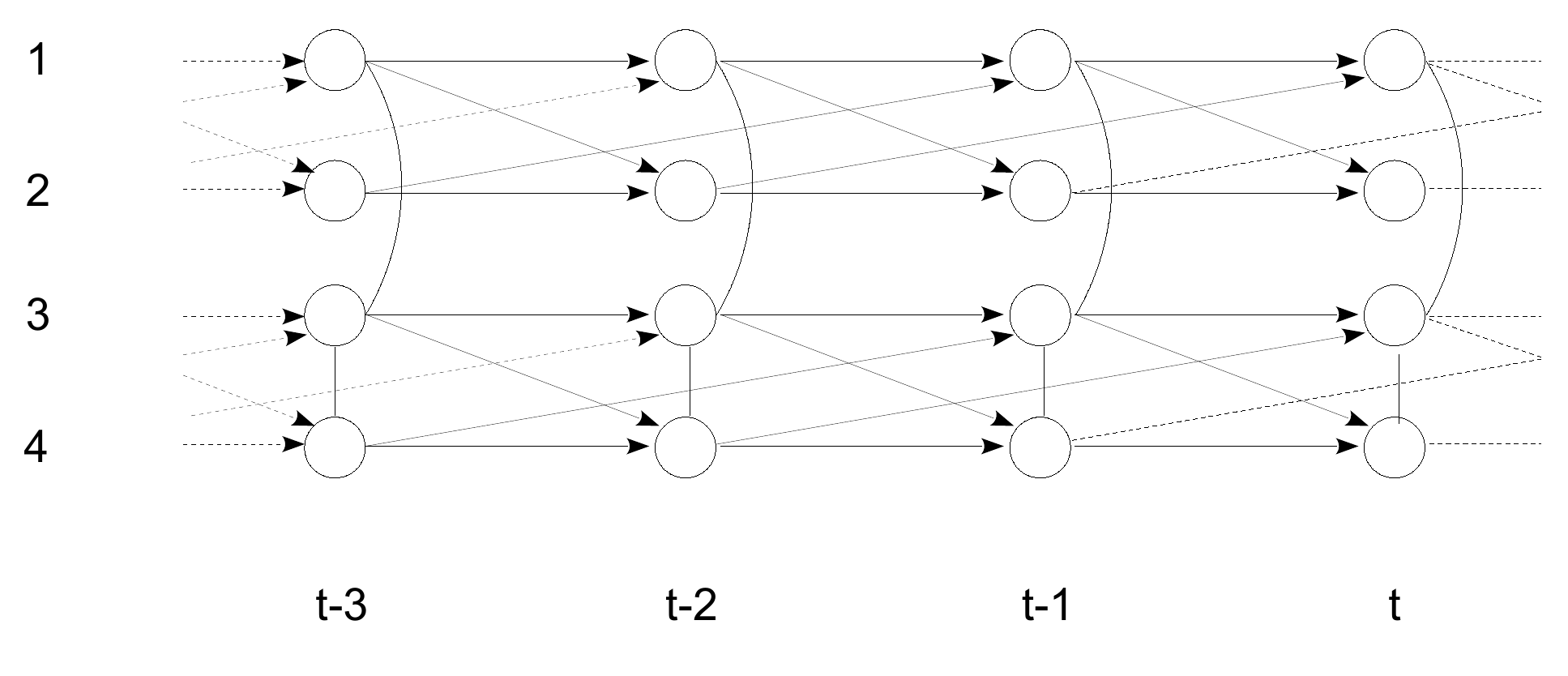}
        \textbf{(a)} TSC graph.

        \includegraphics[width=0.48\textwidth, trim={75 130 210 100}, clip]
        {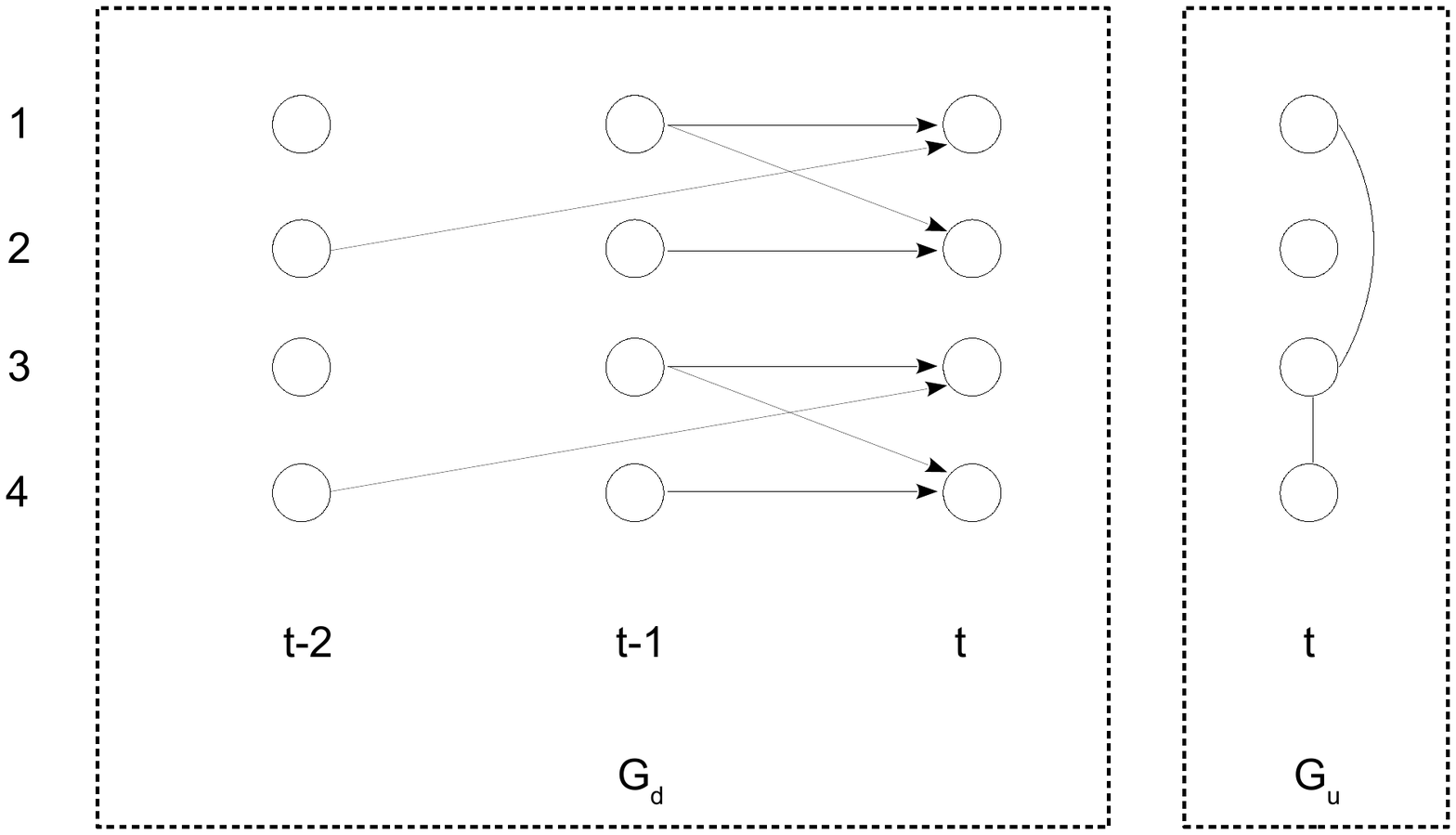}

        \textbf{(b)} GVAR structure.
        \caption{(a) TSC graph and (b) corresponding GVAR structure of a sparse VAR(2) process.
            \label{fig:mixed_graph}}
    \end{center}
\end{figure}

\noindent To compactly encode $G_{TS}$, which essentially is a graph over all
time steps, we will use the subgraphs $G_d \triangleq  (V_d, E_d)$ and $G_u
\triangleq (V_u, E_u)$. The temporal graph $G_d$, with respect to the reference
time step $t$, contains the directed edges from the lagged node sets
$V_{t-1},\ldots,V_{t-k}$ to $V_t$. The contemporaneous graph $G_u$ contains the
undirected edges between the nodes in the same time step; $V_u = V_t$. Finally,
following the graphical VAR modeling framework of \cite{Epskamp-2018}, we refer
to the pair $G \triangleq (G_d, G_u)$ as a graphical VAR (GVAR) structure. For
an example of a GVAR structure, see Figure \ref{fig:mixed_graph}(b).

%===============================================================================

\section{Methods}
\label{section:methods}

The PLVAR method proposed in this paper utilizes a two-step approach for
learning the GVAR model structure. Each step uses a score-based approach, where a
candidate structure is scored according to how well it explains the dependence
structure observed in the data. The proposed scoring function is based on a
Bayesian framework
for model selection. More specifically, PLVAR is based on the concept of
objective Bayes factors for Gaussian graphical models, an approach that has
been used successfully in the cross-sectional setting for learning the
structure of decomposable undirected models \cite{Carvalho-2009}, directed
acyclic graph (DAG) models \cite{Consonni-2012}, and general undirected models
\cite{Leppa-aho-2017}. More recently, this technique has also been used for learning the contemporaneous structure $G_u$ of a graphical VAR model under the assumption that the graph is chordal (decomposable) \cite{Paci-2020}. By utilizing the Bayesian scoring function, we here develop an algorithm that infers the complete model structure $(G_d , G_u)$, where $G_u$ is allowed to be non-chordal. The method requires very little input from the user and is proven to be consistent in the large sample limit.

Given a scoring function, we still need an efficient search algorithm for
finding the optimal structure. To enable scalability to high-dimensional
systems, we propose a search algorithm that exploits the structure of the
considered scoring function via a divide-and-conquer type approach. More
specifically, the initial global structure learning problem is split up into
multiple node-wise sub-problems whose solutions are combined into a final
global structure. Since the number of possible structures in each sub-problem
is still too vast for
finding the global optimum for even moderate values of $d$ and $k$, a greedy
search algorithm is used to find a local maximum in a reasonable time.

\subsection{Bayesian model selection}

The scoring function of the PLVAR method is derived using a Bayesian approach.
Ideally, we are interested in finding a GVAR structure $G$ that is close to optimal in the sense
of maximum a posteriori (MAP) estimation. By the Bayes' rule,
\begin{equation} \label{eqn:bayes_rule}
    p(G | \Y) = \frac{p(\Y | G) p(G)}{\sum p(\Y|G) p(G)}
\end{equation}
where $\Y \triangleq [\yvec_1 \: \cdots \:
\yvec_N]^\top$ is the observation matrix, and the sum in the denominator is
taken over all the possible GVAR structures. Thus, the MAP estimate is given by
\begin{equation}
\hat{G}_{MAP} = \argmax p(\Y|G) p(G).
\end{equation}
A key part of of the above score is the marginal
likelihood $p(\Y|G)$ which is obtained by integrating out the model parameters
$\thetavec$ in the likelihood $p(\Y|G,\thetavec)$ under some prior on the
parameters.

Objective Bayesian model selection poses a problem in the specification of the
parameter prior, as uninformative priors are typically improper. A solution to
this is given by the fractional Bayes factor (FBF) approach \cite{OHagan-1995},
where a fraction of the likelihood is used for updating the improper
uninformative prior into a proper fractional prior. The proper fractional prior
is then used with the remaining part of likelihood to give the FBF. It has been
shown that the FBF provides better robustness against misspecified priors
compared to the use of proper priors in a situation where only vague
information is available about the modeled system \cite{OHagan-1995}.

In this work, we make use of the objective FBF approach for Gaussian graphical
models \cite{Carvalho-2009,Consonni-2012,Leppa-aho-2017,Paci-2020}. In particular, our
procedure is inspired by the fractional marginal pseudo-likelihood (FMPL)
\cite{Leppa-aho-2017}, which allows for non-chordal undirected graphs when
inferring the contemporaneous part of the network.

\subsection{Fractional marginal pseudo-likelihood\label{sec:fmpl}}

In terms of Gaussian graphical models, the FBF was introduced for decomposable
undirected models by \cite{Carvalho-2009}, and extended to DAG models by
\cite{Consonni-2012}, who referred to the resulting score as the fractional
marginal likelihood (FML). Finally, \cite{Leppa-aho-2017} extended the
applicability of the FBF approach to the general class of non-decomposable
undirected models by combining it with the pseudo-likelihood approximation
\cite{Besag-1975}. Our procedure is inspired by the FMPL and can be thought of
as an extension of the approach to time-series data.

As stated in \cite{Leppa-aho-2017}, the FMPL of a Gaussian graphical model is
given by
\begin{equation} \label{eqn:fmpl_dag}
	\begin{aligned}
    pl(\Y | G) &= \prod_{i = 1}^{d} p(\Y_i | \Y_{\text{mb}(i)})\\
    &= \prod_{i = 1}^{d} \pi^{-\frac{N-1}{2}}
    \frac{ \Gamma\left(\frac{N+p_i}{2}\right) }
         { \Gamma\left(\frac{p_i+1}{2}\right) }
    N^{-\frac{2p_i+1}{2}}
    \left(
        \frac{ \left|\mathbf{S_\text{fa($i$)}}\right| }
             { \left|\mathbf{S_\text{mb($i$)}}\right| }
    \right)^{-\frac{N-1}{2}},\\
	\end{aligned}
\end{equation}
where $\text{mb}(i)$ denotes the Markov blanket of $i$, $p_i$ denotes the number of nodes in $\text{mb}(i)$, and $\text{fa}(i)
=\text{mb}(i)\cup i$ is the family of $i$. Moreover, $N$ denotes the number of
observations, $\mathbf{S} = \Y^\top \Y$, and $\mathbf{S_\text{fa($i$)}}$ and
$\mathbf{S_\text{mb($i$)}}$ denote the submatrices of $\mathbf{S}$ restricted
to the nodes that belong to the family and Markov blanket,
respectively, of node $i$. The FMPL is well-defined if $N-1 \geq \max_i (p_i)$, since it ensures (with probability one) positive definiteness of $\mathbf{S_\text{fa($i$)}}$ for each $i= 1,\ldots, d$. Thus, we assume that this (rather reasonable) sparsity condition holds in the following.

Whereas \cite{Leppa-aho-2017} is solely concerned with the learning of
undirected Gaussian graphical models, we are interested in learning a GVAR structure which is made up of both a directed temporal part $G_d$ and an
undirected contemporaneous part $G_u$. Rather than inferring these simultaneously, we break up the problem into two consecutive learning problems that both can be approached using the FMPL. First, in terms of $G_d$, we apply a modified
version of \eqref{eqn:fmpl_dag} where the Markov blankets are built from nodes
in earlier time steps. In terms of $G_u$, we use the inferred directed
structure $\hat{G}_d$ to account for the temporal dependence between the
observations. As this reduces the problem of inferring $G_u$ to the
cross-sectional setting considered in \cite{Leppa-aho-2017}, the score in
\eqref{eqn:fmpl_dag} can be applied as such.

\subsection{Learning the temporal network structure}
\label{subsection:learning_temporal}

In order to learn the temporal part of the GVAR structure, we begin by
creating a lagged data matrix
\begin{equation}\label{eq:Zmatrix}
\mathbf{Z} \triangleq [\Y_{0} \: \Y_{-1} \:
\cdots \: \Y_{-k}]
\end{equation}
where $\Y_{-m} = [\yvec_{k-m+1} \: \cdots \:
\yvec_{N-m}]^\top$. The lagged data matrix thus contains $(k+1)d$ columns, each
corresponding to a time-specific variable; more specifically, the variable of
the $m$:th time series at lag $l$ corresponds to column $ld+m$. As a result of
the missing lag information in the beginning of the time-series, the effective
number of observations (rows) is reduced from $N$ to $N-k$; this, however,
poses no problems as long as $N$ is clearly larger than $k$.

Now, the problem of inferring the temporal structure can be thought of as
follows. For each node $i$, we want to identify, among the lagged variables, a
minimal set of nodes that will shield node $i$ from the remaining lagged variables,
that is, a type of temporal Markov blanket. As seen in
(\ref{eqn:temporal_markov_blanket}), the temporal Markov blanket will equal the
set of lagged variables from which there is a directed edge to node $i$,
commonly referred to as the parents of $i$. Assuming a structure prior that
factorizes over the Markov blankets and a given lag length $k$, we can frame
this problem as the following optimization problem based on the lagged data
matrix and a modified version of \eqref{eqn:fmpl_dag}:
\begin{equation} \label{eqn:fmpl_temp_opt}
	\begin{aligned}
    &\underset{\text{mb}(i)}{\arg\max} \ \ \ \prod_{i = 1}^{d} p(\mathbf{Z}_i | \mathbf{Z}_{\text{mb}(i)}, k)p(\text{mb}(i)|k)\\
    &\text{subject to } \ \ \text{mb}(i) \subseteq \{ d+1,\ldots,(k+1)d \},
	\end{aligned}
\end{equation}
where the so-called local FMPL is here defined as
\begin{align} \label{eqn:local_fmpl_temp}
    & p(\mathbf{Z}_i | \mathbf{Z}_{\text{mb}(i)}, k) \triangleq \nonumber \\
    & \pi^{-\frac{N-k-1}{2}}
    \frac{ \Gamma\left(\frac{N-k+p_i}{2}\right) }
         { \Gamma\left(\frac{p_i+1}{2}\right) }
    (N-k)^{-\frac{2p_i+1}{2}}
    \left(
        \frac{ \left|\mathbf{S_\text{fa($i$)}}\right| }
             { \left|\mathbf{S_\text{mb($i$)}}\right| }
    \right)^{-\frac{N-k-1}{2}}.
\end{align}
The unscaled covariance matrix in \eqref{eqn:local_fmpl_temp} is now obtained
from the lagged data matrix $\mathbf{Z}$, that is, $\mathbf{S} =
\mathbf{Z}^\top \mathbf{Z}$.

The only thing left to specify in \eqref{eqn:fmpl_temp_opt} is the structure
prior. To further promote sparsity in the estimated network, we draw
inspiration from the extended Bayesian information criterion (eBIC; see, for
example,
\cite{Barber-2015})
and define our prior as
\begin{equation} \label{eqn:prior}
    p(G_d|k) = \prod_{i=1}^{d} p(\text{mb}(i)|k) = \prod_{i=1}^{d}(kd)^{- \gamma p_i},
\end{equation}
where $p_i$ is the number of
nodes in the Markov blanket of node $i$ and $\gamma$ is a tuning parameter adjusting the strength of the sparsity promoting effect. As the default value
for the tuning parameter, we use $\gamma = 0.5$. As seen later, this prior will
be essential for selecting the lag length.

From a computational point of view, problem \eqref{eqn:fmpl_temp_opt} is very
convenient in the sense that it is composed of $d$ separate sub-problems that
can be solved independently of each other. To this end, we apply the same
greedy search algorithm as in \cite{Pensar-2017,Leppa-aho-2017}.

\subsection{Search algorithm}

The PLVAR method uses a greedy search algorithm to find the set $\text{mb}(i)$ (with $p_i \leq N-1$)
that maximizes the local FMPL for each $i \in V$. The search algorithm used by the PLVAR method has been published as Algorithm 1
in \cite{Leppa-aho-2017} and originally introduced as Algorithm 1 in
\cite{Pensar-2017}, based on \cite{Tsamardinos-2006}. This algorithm works in
two interleaving steps, deterministically adding and removing individual nodes
in the set $\text{mb}(i)$. At each step, the node $j \notin \text{mb}(i)$ that
yields the largest increase in the local FMPL is added to or removed from
$\text{mb}(i)$. This is a hybrid approach aimed for low computational
cost without losing much of the accuracy of an exhaustive search
\cite{Pensar-2017}. Pseudo-code for the algorithm is given in Appendix
A.

\subsection{Selecting the lag length}

As noted in Subsection \ref{subsection:learning_temporal}, the FMPL of a
graphical VAR model depends on the lag length $k$. Unfortunately, the true
value of $k$ is usually unknown in real-world problems. For the PLVAR method,
we adopt the approach of setting an upper bound $K$ and learning the temporal
structure separately for each $k = 1, ..., K$. The value of $K$ should
preferably be based on knowledge about the problem domain, e.g. by considering
the time difference between the observations and the maximum time a variable
can influence itself or the other variables.

Once the upper bound $K$ has been set, a score is calculated for each $k$.
Typical scores used for this purpose include the Bayesian information criterion
and the Akaike information criterion, both of which are to be minimized. For
the PLVAR method, we eliminate the need for external scores by selecting the $k$
that maximizes our FMPL-based objective function in \eqref{eqn:fmpl_temp_opt}.
To ensure comparable data sets for different values on $k$, the time steps
included in the lagged data matrices are determined by the maximum lag $K$.
Note that our lag selection criterion is made possible by our choice of
structure prior \eqref{eqn:prior} which depends on $k$. As there is no
restriction that a temporal Markov blanket should include nodes of all
considered lags, the FMPL part of the objective function does not explicitly
penalize overly long lag lengths.

\subsection{Learning the contemporaneous network structure}

Once the temporal network structure has been estimated, we proceed to learn the
contemporaneous structure in three steps. First, the non-zero parameters of the
lag matrices $\mathbf{A}_m$ are estimated via the ordinary least squares
method (which is applicable under the sparsity condition $p_i \leq N-1$): this is done separately for each node given its temporal Markov blanket (or parents) in the temporal
network. Next, the fully estimated lag matrices $\hat{\mathbf{A}}_m$ are used
to calculate the residuals $\hat{\epsvec}_t$ between the true and the predicted
observations:
\begin{equation}
    \hat{\epsvec}_t = \yvec_t - \sum_{m = 1}^{k} \hat{\mathbf{A}}_m \yvec_{t-m}.
\end{equation}
If the estimates $\hat{\mathbf{A}}_m$ are accurate, the residuals
$\hat{\epsvec}_t$ will form a set of approximately cross-sectional data,
generated by the contemporaneous part of the model, and the problem is reduced
to the setting of the original FMPL method. Thus, as our objective function, we
use the standard FMPL \eqref{eqn:fmpl_dag} in combination with our eBIC type
prior:
\begin{equation} \label{eqn:prior2}
    p(G_u) = \prod_{i=1}^{d} p(\text{mb}(i)) = \prod_{i=1}^{d}(d-1)^{- \gamma p_i},
\end{equation}
The only difference to \eqref{eqn:prior} is the number of possible Markov blanket members which is now $d-1$. Again, as the default value
for the tuning parameter, we use $\gamma = 0.5$.

Although the Markov blankets are now connected through the symmetry constraint
$i \in \text{mb}(j) \Leftrightarrow j \in \text{mb}(i),$ we use the same
divide-and-conquer search approach as was used for the temporal part of the
network. However, since the search may then result in asymmetries in the Markov
blankets, we use the OR-criterion as described in \cite{Leppa-aho-2017} when
constructing the final graph: if $i \in
\text{mb}(j)$ or $j \in \text{mb}(i)$, then $(i,j) \in E_u$.

\subsection{Consistency of PLVAR}

A particularly desirable characteristic for any statistical estimator is
consistency. This naturally also applies when learning the GVAR structure. Importantly,
it turns out that the proposed PLVAR method enjoys consistency. More
specifically, as the sample size tends to infinity, PLVAR will recover the true
GVAR structure, including the true lag length $k^*$, if the maximum considered
lag length $K$ is larger than or equal to $k^*$. A brief outline of the
rationale is provided in this subsection, and a formal proof of the key steps
is provided in Appendix B.

The authors of \cite{Leppa-aho-2017} show that their FMPL score in combination
with the greedy search algorithm is consistent for finding the Markov blankets
in Gaussian graphical models. Assuming a correct temporal structure, the lag
matrices can be consistently estimated by the LS method \cite{Lutkepohl-2005},
and the residuals will thus tend towards i.i.d. samples from $N(\zerovec,
\boldsymbol{\Sigma})$, where the sparsity pattern of $\Omega = \Sigma^{-1}$
corresponds to the contemporaneous structure. This is the setting considered in
\cite{Leppa-aho-2017}, for which consistency has been proven.

What remains to be shown is that the first step of PLVAR is consistent in terms
of
recovering the temporal structure. Since PLVAR uses a variant of FMPL where the
set $\text{mb}(i)$ is selected from lagged variables, it can be proven, using a
similar approach as in \cite{Leppa-aho-2017}, that PLVAR will recover the true
temporal Markov blankets (or parents) in the large sample limit if $K \geq k^*$
(see Appendix B for more details). Finally, since the structure prior
\eqref{eqn:prior} is designed to favor shorter lag lengths, it will in
combination with the FMPL ensure that the estimated lag length will equal $k^*$
in the large sample limit.

As opposed to LASSO and SCAD, PLVAR needs no conditions on the measurement matrix and relation to true signal strength for model selection consistency \cite{Zhao-2006,Buhlmann-2013,Huang-2007}. The reason for this lies in the difference between the approaches. In LASSO and SCAD, the regularization term introduces bias to the cost function, thereby deviating the estimation from the true signal \cite{Zhao-2006,Buhlmann-2013}. On the other hand, PLVAR estimates the graph structure using the FMPL score: this is a Bayesian approach and thus has no bias in asymptotic behavior \cite{Consonni-2012,Leppa-aho-2017}.

\subsection{Estimating the remaining parameters}

Once the PLVAR method has been used to infer the model structure, the remaining
parameters can be estimated as described in this subsection. This combination
provides a complete pipeline for the learning of graphical VAR models.

The parameters that still remain unknown at this point are the non-zero
elements of the lag matrices $\mathbf{A}_m, m = 1, \ldots, k$ and the non-zero
elements of the precision matrix $\boldsymbol{\Omega}$. We calculate ML estimates
for these parameters iteratively until a convergence criterion is met.

In each iteration, we first estimate the remaining elements of $\mathbf{A}_m$ given the current precision matrix. As initial value for the precision matrix, we use the identity matrix.
In order to enforce the sparsity pattern implied by the temporal network, we use the ML estimation procedure described in
\cite{Lutkepohl-2005}. Next, we calculate the residuals between the true and predicted observations. We then estimate the remaining elements of
$\boldsymbol{\Omega}$ from the residuals while enforcing the sparsity pattern learned earlier for
the contemporaneous network: this is done by applying the ML estimation
approach outlined in \cite{Hastie-2009}.

The iterations are repeated until the absolute difference between the current and the
previous maximized log-likelihood of $\boldsymbol{\Omega}$ is smaller than a
threshold $\delta$.

%===============================================================================

\section{Experiments and results}
\label{section:experiments_and_results}

In order to assess the performance of the PLVAR method, we made a number of
experiments on both synthetic and real-world data. For the synthetic data, we
used a single node in the Puhti supercomputer of CSC – IT Center for Science,
Finland, equipped with 4 GB RAM and an Intel\textsuperscript{\textregistered}
Xeon\textsuperscript{\textregistered} Gold 6230 processor with 20
physical cores at 2.10 GHz base frequency. For the real-world data, we used a
laptop workstation equipped with 8 GB RAM and an
Intel\textsuperscript{\textregistered} Core\textsuperscript{\texttrademark}
i7-6700HQ CPU with 4 physical cores at 2.60 GHz base frequency. The algorithms
were implemented in the R language. For the sparsity-constrained ML estimation
of $\boldsymbol{\Omega}$, we used the implementation in the R package mixggm
\cite{Fop-2019}. The threshold $\delta$ for determining the convergence of
the parameter estimation was set to $10^{-6}$.

\subsection{Synthetic data}

The performance of the PLVAR method was first evaluated in a controlled setting on data generated from known models with a ground truth structure. Having access to the true network structure, the quality of the inferred network structures were assessed by \emph{precision}: the proportion of true edges among the inferred edges, and \emph{recall}: the proportion of detected edges out of all true edges. Precision and recall were calculated separately for the temporal and the contemporaneous part of the network.

The R package \texttt{SparseTSCGM}\footnote{More specifically, we used the function \texttt{sim.data()}; see \cite{Abegaz-2013} for more details regarding the simulator.} \cite{Abegaz-2013} was used to generate GVAR models with a random network structure and lag length $k=2$. 
We performed two experiments. In the first experiment, we fixed the degree of sparsity and varied the number of variables, $d\in\{20,40,80 \}$. In the second experiment, we fixed the number of variables to $d = 40$ and varied the degree of sparsity. The degree of sparsity in the generated networks is controlled by an edge inclusion probability parameter, called $\texttt{prob0}$. 
In order to have similar adjacency patterns for different number of variables, we specified the edge inclusion parameter as $q/(kd)$, meaning that the average indegree in the temporal structure is $q$ and the average number of neighbors in the contemporaneous structure is $q/2$. In the first experiment, we set $q = 3$, resulting in rather sparse structures, and in the second experiment, we let $q\in\{5,7,9 \}$, resulting in increasingly denser structures. For each considered configuration of $d$ and $q$, we generated 20 stable GVAR models. Finally, a single time series over 800
time points was generated from each model, and the $N$ first
observations, with $N$ varying between 50 and 800, were used to form the final
collection of data sets.

For comparison, LASSO and SCAD were also applied to the generated
data, using the implementations available in \texttt{SparseTSCGM}. While LASSO and SCAD were given the correct lag length $k =
2$, PLVAR was given a maximum allowed lag length $K = 5$ and thereby
tasked with the additional challenge of estimating the correct lag length. In terms of the tuning parameters, the
$\gamma$ prior parameter in PLVAR was set to 0.5, and the remaining parameters
of LASSO and SCAD were set according to the example provided in the \texttt{SparseTSCGM} manual. More specifically, the maximum number of outer and inner loop iterations were set to 10 and 100, respectively. The default grid of candidate values for the tuning parameters (\texttt{lam1} and \texttt{lam2}) were used, and the optimal values were selected using the modified Bayesian information criterion (\texttt{bic\_mod}).

The precision, recall and running times of the considered methods are summarized in Figure \ref{fig:res_sim} for the first experiment and in Figure \ref{fig:res_sim_rev} for the second experiment. For some of the most challenging settings: $N=50$ for $d\in\{ 40,80 \}$ and
$N=100$ for $d=80$, LASSO and SCAD exited with an error and the results are
therefore missing. In terms of precision (Figures \ref{fig:res_sim}(a) and \ref{fig:res_sim_rev}(a)), PLVAR is clearly the most accurate of the methods, both for the temporal and the contemporaneous part of the network structure. In terms
of recall (Figures \ref{fig:res_sim}(b) and \ref{fig:res_sim_rev}(b)), all methods perform very well for the larger samples, recovering close to all edges. For smaller sample sizes, PLVAR has a slight edge on both LASSO and SCAD in the sparser settings, $q=3$ and $q=5$, while it is overall quite even in the denser setting $q=7$, and the situation is almost reversed for $q=9$.

As expected, increasing the number of variables or the number of edges makes the learning problem harder, leading to a reduced accuracy. However, there is no drastic reduction in accuracy for any of the methods. The most drastic change between the settings is perhaps the computation time of LASSO and SCAD when the number of variables is increased (Figure \ref{fig:res_sim}(c)). 
Overall, PLVAR was orders of magnitude
faster than LASSO and SCAD, both of which use cross-validation to select the regularization parameters. However, it must be noted that PLVAR solves a simpler problem in that it focuses on learning the model structure (yet without a given lag length), while the likelihood-based methods also estimate the model parameters
along with the structure. Still, the problem of estimating the model (parameters) is
greatly facilitated once the model structure has been inferred. 

Finally, when it comes to selecting the correct lag length ($k=2$) from the
candidate lag lengths, PLVAR was highly accurate and estimated the lag length
to be higher than 2 for only a few cases in both the sparse and the denser settings (Figure
\ref{fig:lag_sim}). It is also
worth noting that while the maximum lag length supported by the PLVAR method is
only limited by the available computing resources, the LASSO and SCAD implementations in \texttt{SparseTSCGM} currently only support lag lengths 1
and 2.

\begin{figure*}
    \begin{center}
        \includegraphics[width=0.9\textwidth, trim={0 0 0 0}, clip]
        {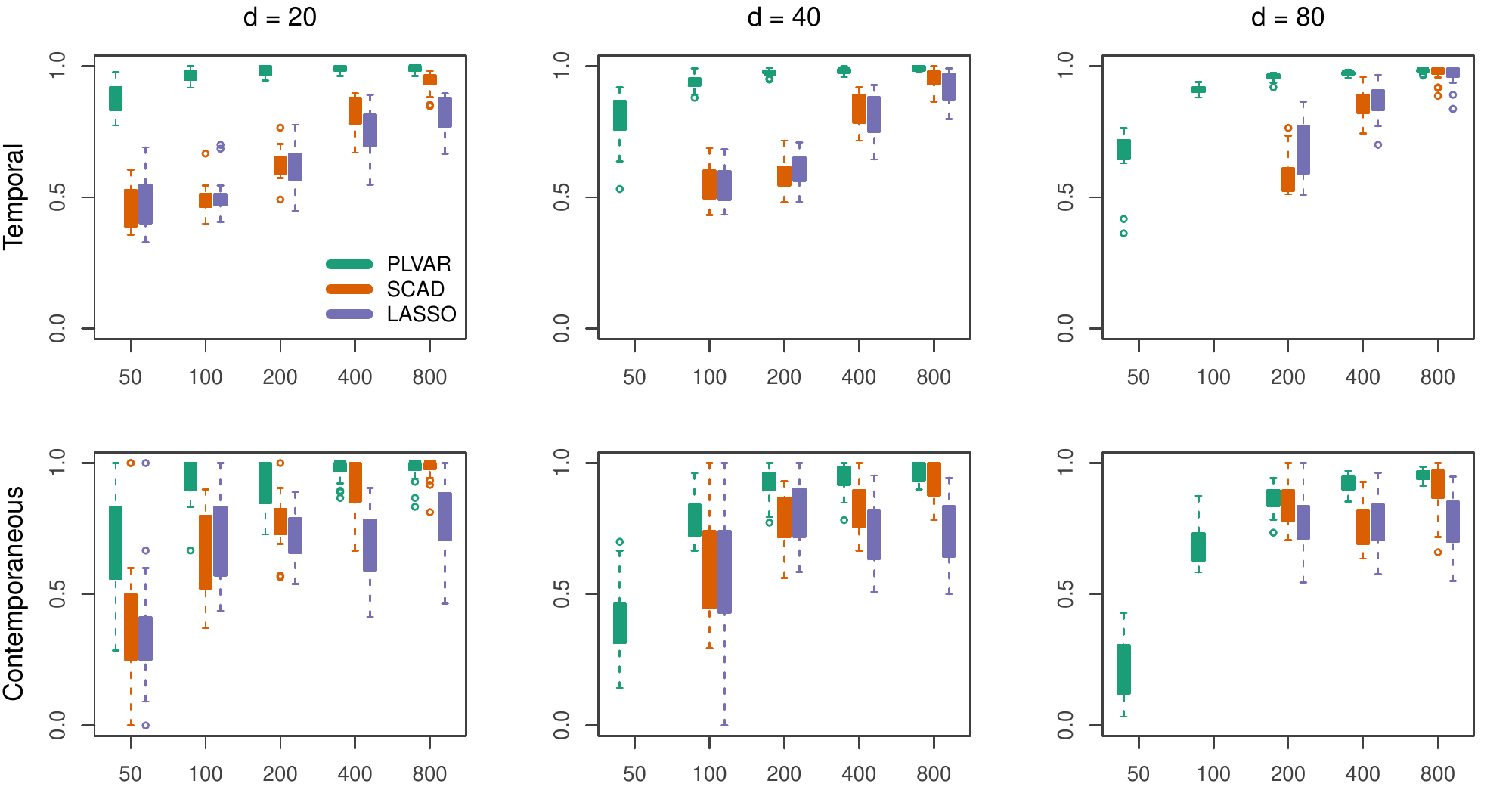}

        \textbf{(a)} Precision. \vspace{0.5cm}

        \includegraphics[width=0.9\textwidth, trim={0 0 0 0}, clip]
        {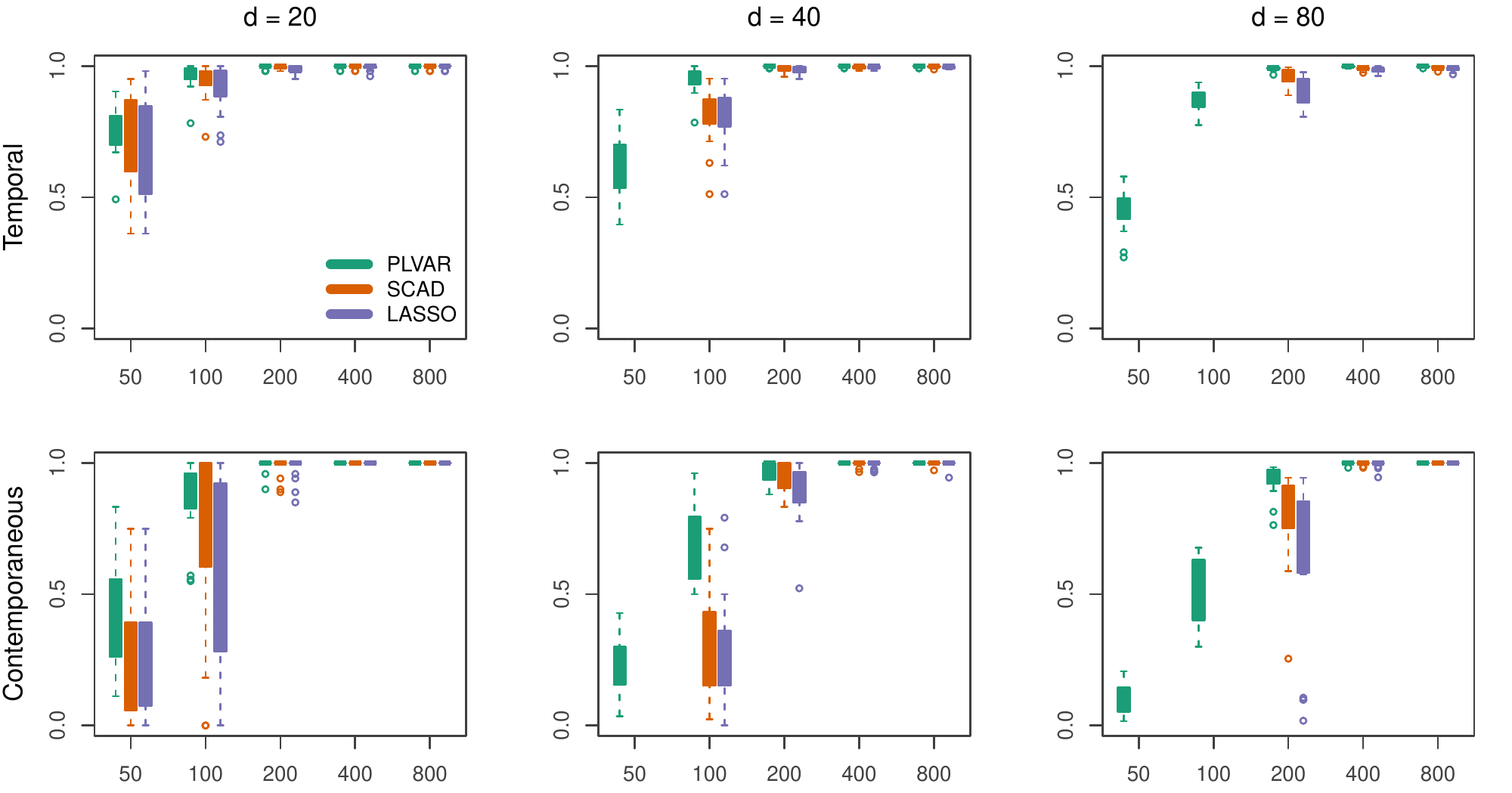}

        \textbf{(b)} Recall. \vspace{0.5cm}

        \includegraphics[width=0.9\textwidth, trim={0 0 0 0}, clip]
        {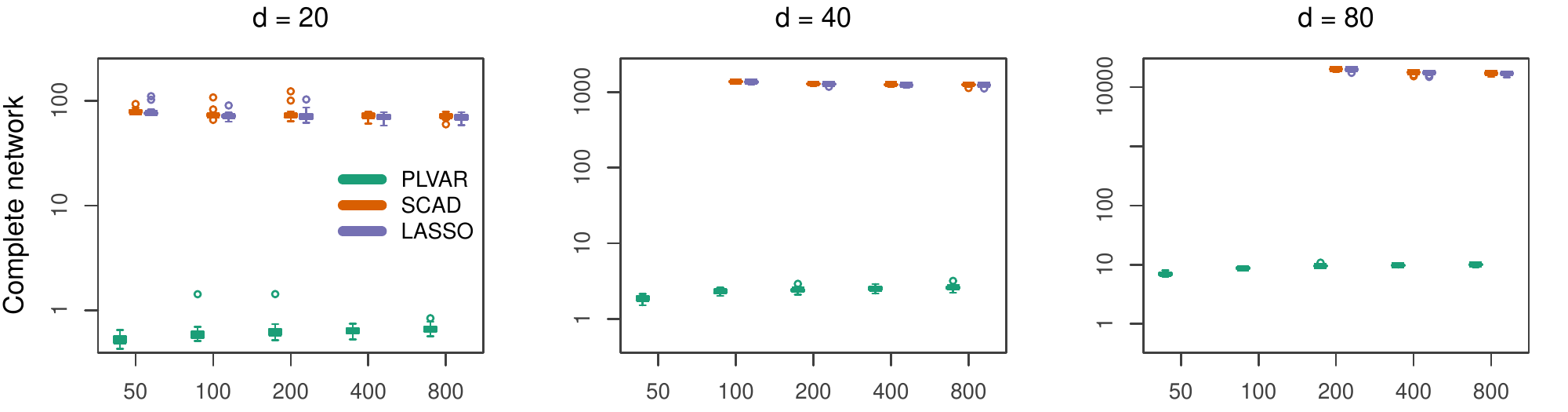}

        \textbf{(c)} Computation time. \vspace{0.1cm}

        \caption{Results of the simulation study with a fixed degree of sparsity ($q = 3$) and an increasing number of variables: (a) Precision and (b) recall
        (y-axis) for the temporal and contemporaneous part of the network
        structure for various sample sizes (x-axis) and model size, $d$. (c)
        Recorded  computation times (y-axis in log scale) for learning the
        complete network structure. Part of the results are missing for LASSO
        and SCAD since they exited with an error for some of the settings. The
        statistics shown in the plots have been calculated from a sample of 20
        random GVAR models.
        \label{fig:res_sim}}
    \end{center}
\end{figure*}

\begin{figure*}
    \begin{center}
        \includegraphics[width=0.9\textwidth, trim={0 0 0 0}, clip]
        {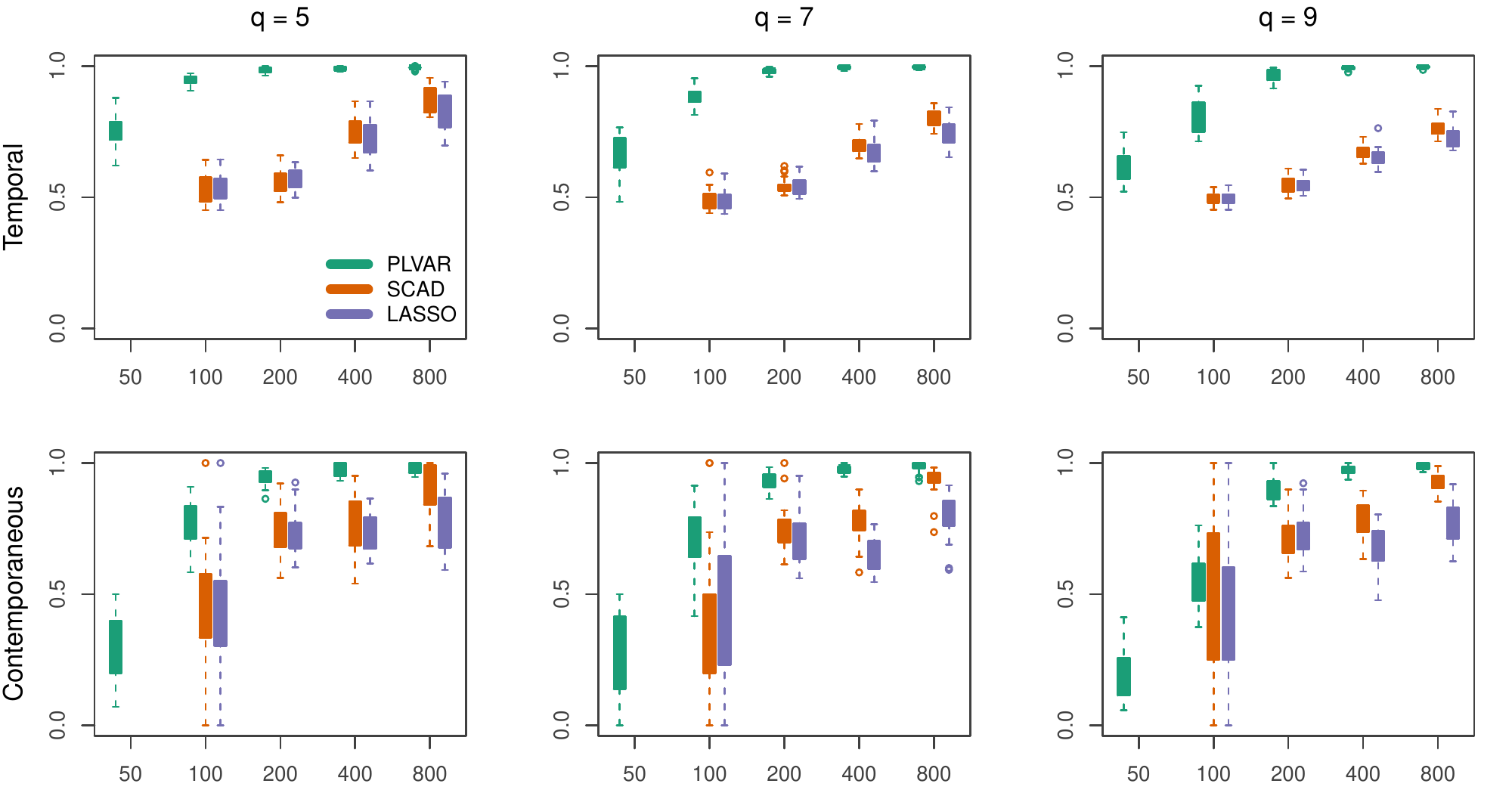}

        \textbf{(a)} Precision. \vspace{0.5cm}

        \includegraphics[width=0.9\textwidth, trim={0 0 0 0}, clip]
        {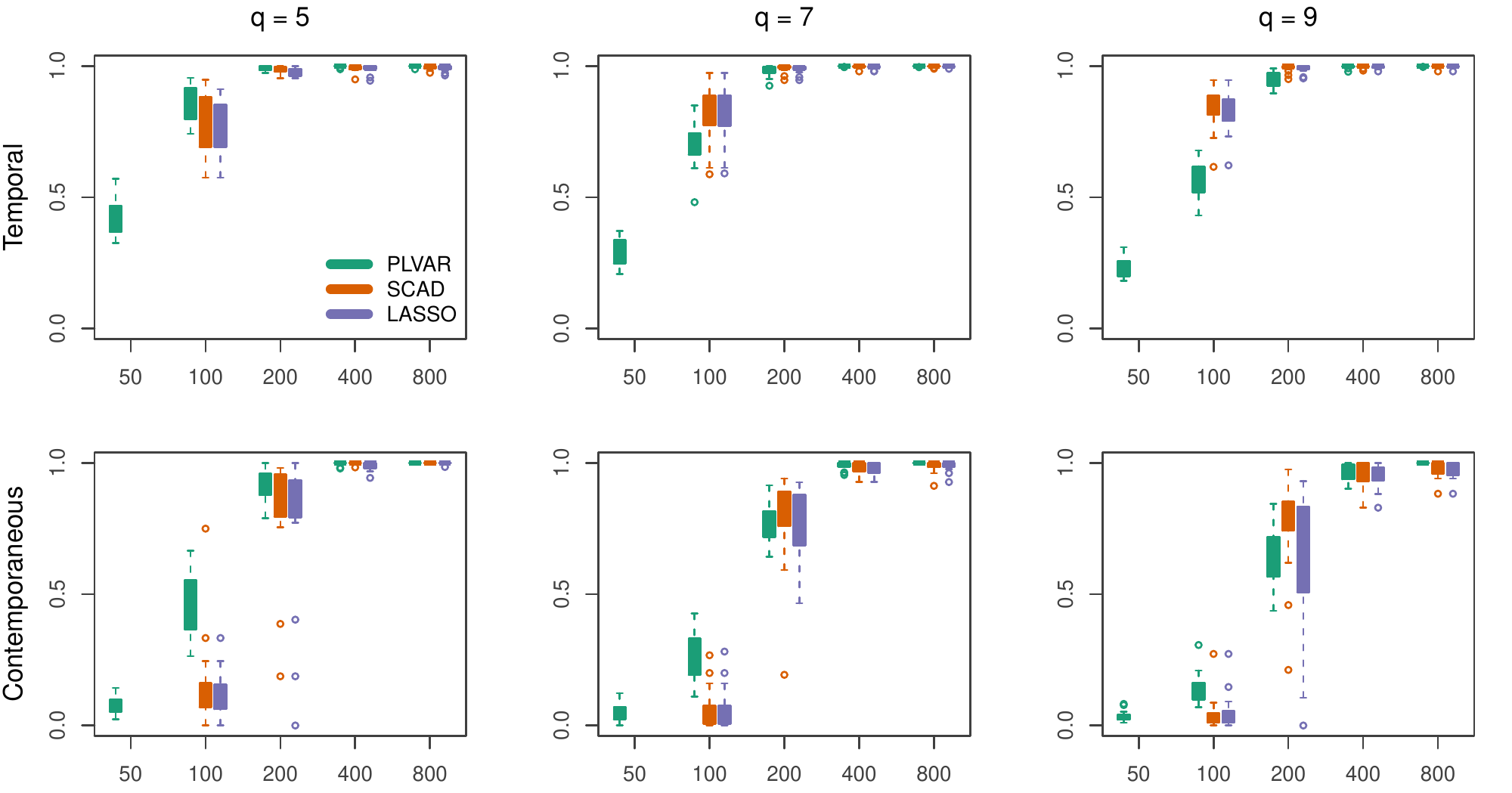}

        \textbf{(b)} Recall. \vspace{0.5cm}

        \includegraphics[width=0.9\textwidth, trim={0 0 0 0}, clip]
        {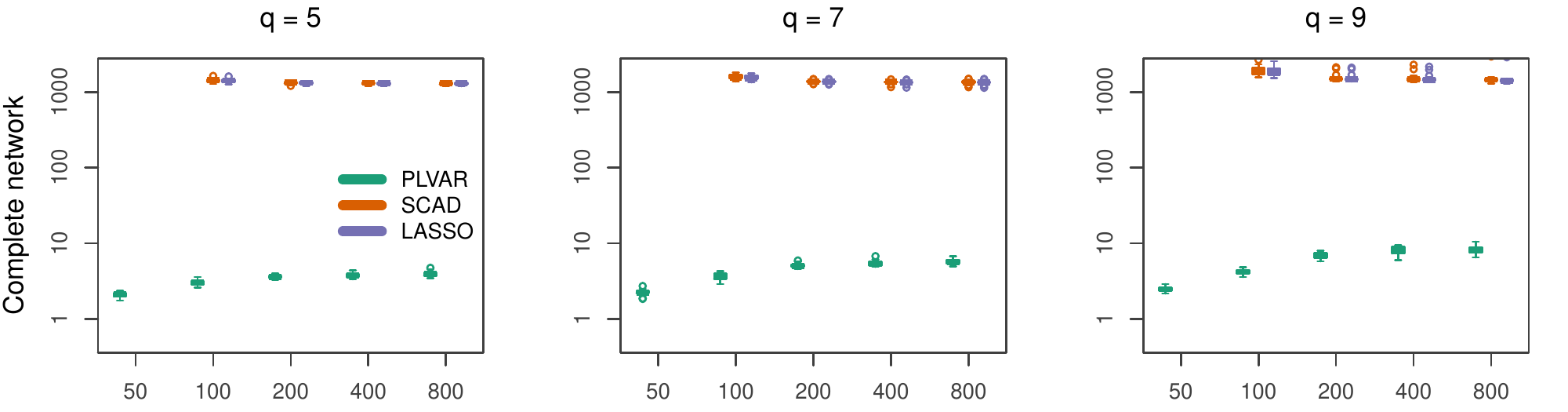}

        \textbf{(c)} Computation time. \vspace{0.1cm}

        \caption{Results of the simulation study with a fixed number of variables ($d=40$) and a decreasing degree of sparsity: (a) Precision and (b) recall
        (y-axis) for the temporal and contemporaneous part of the network
        structure for various sample sizes (x-axis) and degrees of sparsity, as controlled by $q$. (c)
        Recorded  computation times (y-axis in log scale) for learning the
        complete network structure. Part of the results are missing for LASSO
        and SCAD since they exited with an error when $N=50$. The
        statistics shown in the plots have been calculated from a sample of 20
        random GVAR models.
        \label{fig:res_sim_rev}}
    \end{center}
\end{figure*}

\begin{figure*}
    \begin{center}
        \includegraphics[width=0.9\textwidth, trim={0 0 0 0}, clip]
        {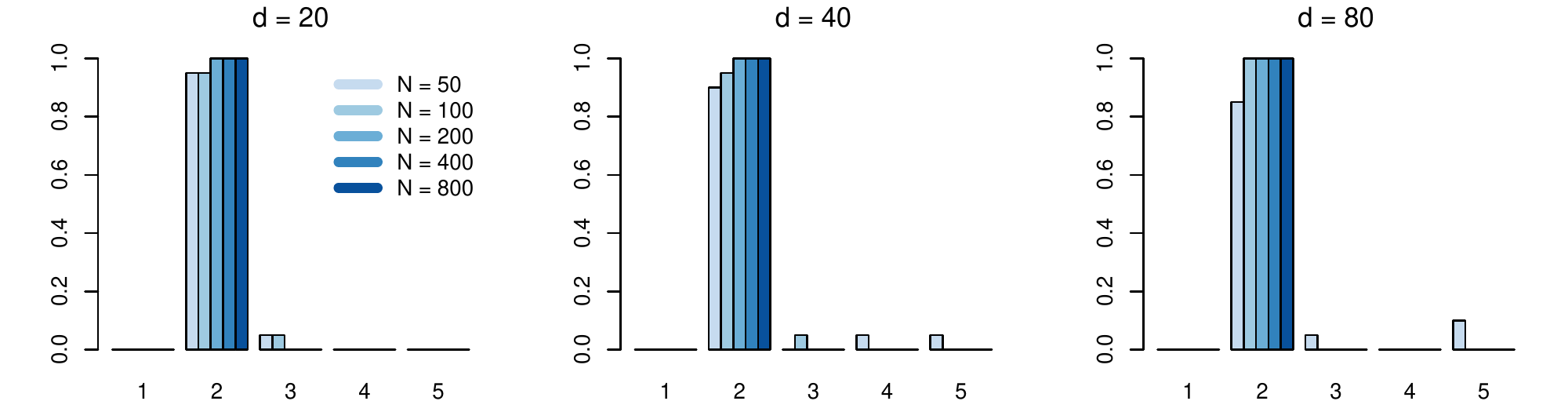}

	\textbf{(a)} \vspace{0.5cm}

        \includegraphics[width=0.9\textwidth, trim={0 0 0 0}, clip]
        {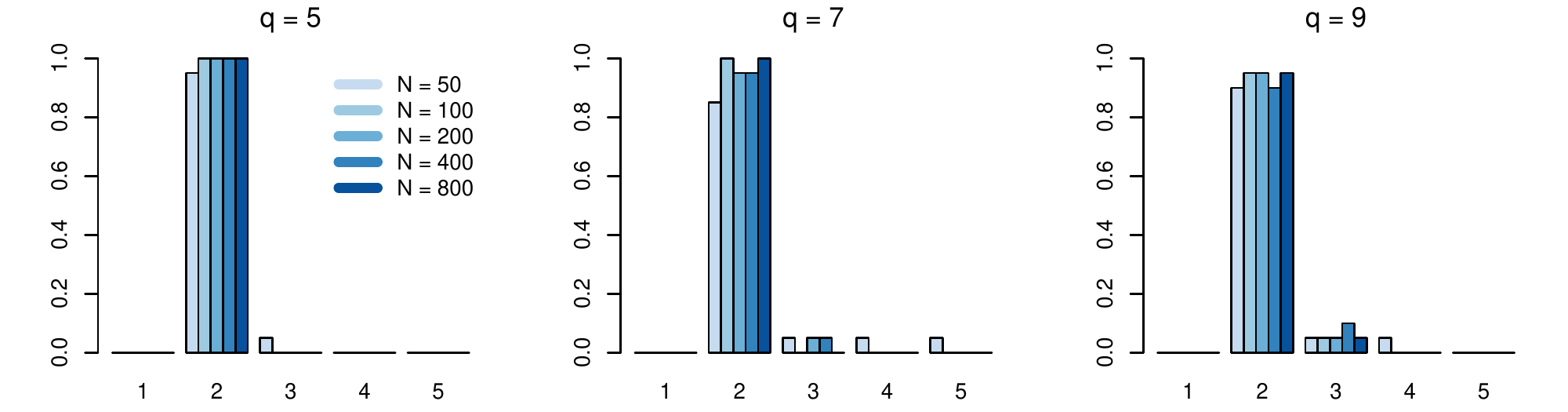}

	\textbf{(b)}

        \caption{Estimated lag length by PLVAR for the different sample sizes
        in the simulation studies: (a) $q = 3$ and (b) $d = 40$. The x-axis represents the candidate lag
        lengths and the y-axis represents the percentage of times a specific
        lag length was chosen. The correct lag length is 2 in all cases.
        \label{fig:lag_sim}}
    \end{center}
\end{figure*}

\subsection{Electroencephalography data}

Since VAR models have found extensive use in EEG analysis, the experiments made
on synthetic data were complemented with experiments on real-world EEG data.
The dataset chosen for this purpose is a subset of the publicly available
CHB-MIT Scalp EEG Database (CHBMIT) \cite{Shoeb-2009,Goldberger-2000}.

The CHBMIT database contains a total of 24 cases of EEG recordings, collected
from 23 pediatric subjects at the Children's Hospital Boston. Each case
comprises multiple recordings stored in European Data Format (.EDF) files.
However, only the first recording of each unique subject was used in order to
speed up the experiments. Case 24 was excluded since it has been added
afterwards and is missing gender and age information. The recordings included
in the experiments are listed in Table \ref{table:recordings_in_chbmit}.

\begin{table}
    \caption{Recordings included in the EEG experiments}
    \label{table:recordings_in_chbmit}
    \centering
    \begin{tabular}{lll}
        \hline\noalign{\smallskip}
        recording filename & subject gender & subject age (years) \\
        \noalign{\smallskip}\hline\noalign{\smallskip}
        chb01\_01.edf & female & 11 \\
        chb02\_01.edf & male & 11 \\
        chb03\_01.edf & female & 14 \\
        chb04\_01.edf & male & 22 \\
        chb05\_01.edf & female & 7 \\
        chb06\_01.edf & female & 1.5 \\
        chb07\_01.edf & female & 14.5 \\
        chb08\_02.edf & male & 3.5 \\
        chb09\_01.edf & female & 10 \\
        chb10\_01.edf & male & 3 \\
        chb11\_01.edf & female & 12 \\
        chb12\_06.edf & female & 2 \\
        chb13\_02.edf & female & 3 \\
        chb14\_01.edf & female & 9 \\
        chb15\_02.edf & male & 16 \\
        chb16\_01.edf & female & 7 \\
        chb17a\_03.edf & female & 12 \\
        chb18\_02.edf & female & 18 \\
        chb19\_02.edf & female & 19 \\
        chb20\_01.edf & female & 6 \\
        chb22\_01.edf & female & 9 \\
        chb23\_06.edf & female & 6 \\
        \noalign{\smallskip}\hline
    \end{tabular}
\end{table}

The recordings in the CHBMIT database have been collected using the
international 10-20 system for the electrode placements. The channels derived
from the electrodes vary between files and some channels only exist in a
subset of the recordings. Hence, the experiments were made using only those 21
channels that can be found in all the recordings, implying $d = 21$ in the VAR
models. Table \ref{table:channels_in_chbmit} lists the channels included in the
experiments.

\begin{table}
    \caption{Channels included in the EEG experiments}
    \label{table:channels_in_chbmit}
    \centering
    \begin{tabular}{llll}
        \hline\noalign{\smallskip}
        FP1-F7 & F7-T7 & T7-P7 & P7-O1 \\
        FP1-F3 & F3-C3 & C3-P3 & P3-O1 \\
        FP2-F4 & F4-C4 & C4-P4 & P4-O2 \\
        FZ-CZ  & CZ-PZ & T7-FT9 &      \\
        FT9-FT10 & FT10-T8 &    &      \\
        \noalign{\smallskip}\hline
    \end{tabular}
\end{table}

All the recordings have a sampling rate of 256 Hz. The first 60 seconds were
skipped in each recording because the signal-to-noise ratio in the beginning of
some recordings is very low. The next 256 samples were used as training data
and presented to each method for learning the model structure and the model
parameters. The following 512 samples were used as test data to evaluate
the performance of the methods.

Each experiment was made using 6 different algorithms. In addition to
PLVAR, LASSO, and SCAD, conventional LS estimation was included for comparison.
As noted before, the SparseTSCGM implementations of LASSO and SCAD only allow
lag lengths 1 and 2. Therefore, the final set of algorithms comprised
PLVAR, LS, PLVAR2, LS2, LASSO2, and SCAD2. In this set, algorithm names ending
with 2 refer to the respective methods provided with lag length 2 for
straightforward comparison. The names PLVAR and LS refer to respective methods
with no a priori lag length information.

\begin{figure*}
    \centering
    \includegraphics[width=.88\textwidth,trim={0 0.5cm 0
        0},clip]{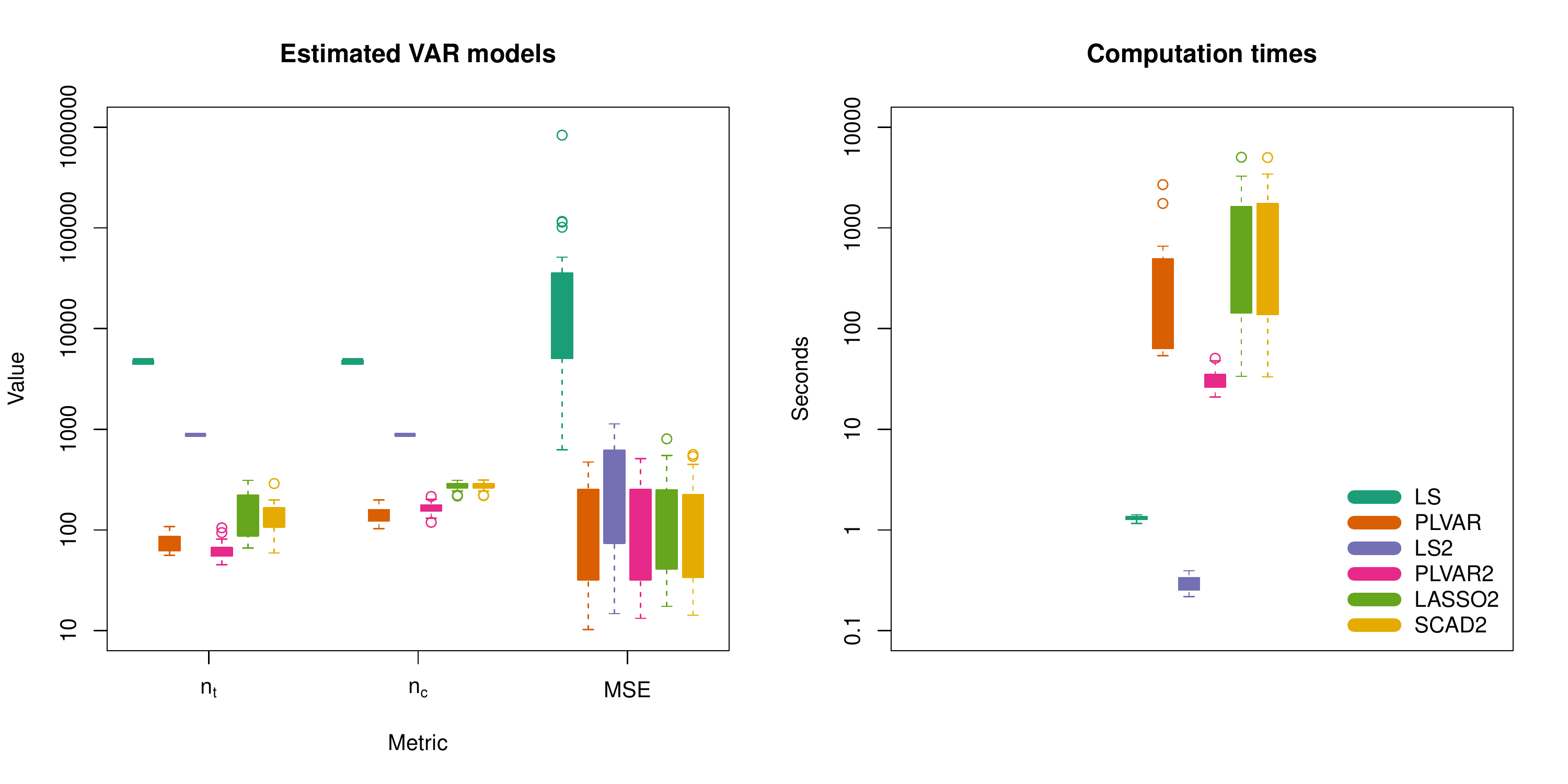}
    \caption{Results on EEG data. $n_\text{t}$, $n_\text{c}$, and MSE denote
        the number of edges in the temporal network, the number of edges in the
        contemporaneous network, and the mean squared error of 1-step
        predictions
        on the test data, respectively. Note the logarithmic $y$ axis in both
        subplots.}
    \label{figure:results_chbmit_1_min_offset}
\end{figure*}

\begin{figure*}
    \centering
    \includegraphics[width=.58\textwidth]
    {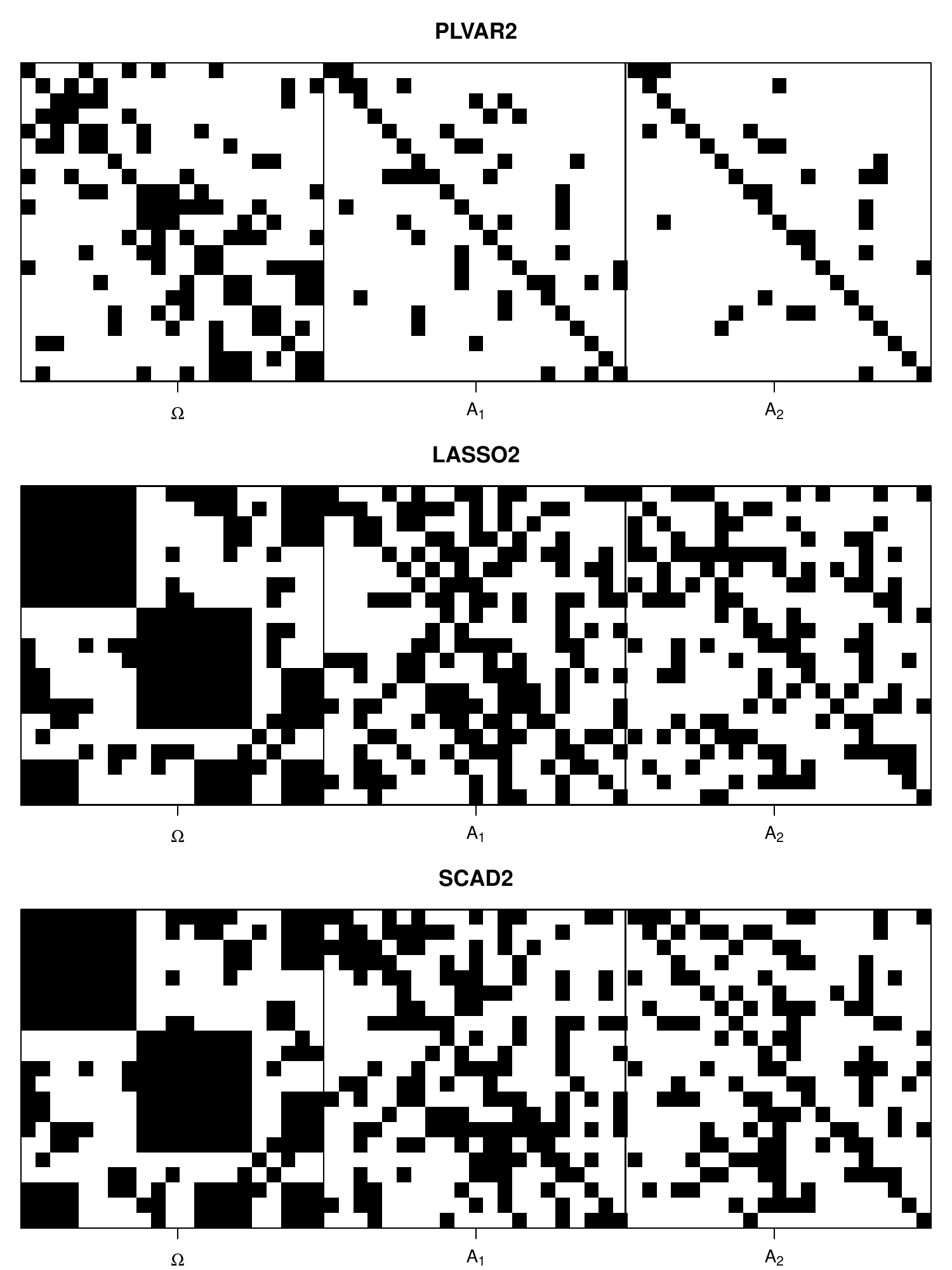}
    \caption{Adjacency matrix estimates, obtained from the EEG recording
        chb06\_01.edf using the PLVAR2, the LASSO2, and the SCAD2 methods. The
        black and white pixels correspond to non-zero and zero elements in the
        matrices, respectively. $\boldsymbol{\Omega}$ refers to the precision
        matrix, $\mathbf{A}_1$ to the first lag matrix, and $\mathbf{A}_2$ to
        the second lag matrix. Note the sparsity of the PLVAR2 estimates as
        compared to the other methods.}
    \label{figure:results_matrix_visualization}
\end{figure*}

The performance of each algorithm was evaluated using the following metrics:
the number of edges in the estimated temporal network ($n_\text{t}$), the
number of edges in the estimated contemporaneous network ($n_\text{c}$), the
mean squared error (MSE) of 1-step predictions on the test data, and the
estimation time in seconds ($t$). The results are summarized in Figure
\ref{figure:results_chbmit_1_min_offset}. An example of the estimated temporal
and contemporaneous networks is given in Figure
\ref{figure:results_matrix_visualization} where the networks are visualized
through their adjacency matrices.

Figure \ref{figure:results_chbmit_1_min_offset} shows that PLVAR and PLVAR2
produced the sparsest estimates for both the temporal network and the
contemporaneous network. The dense estimates produced by LS and LS2 had an
order of magnitude more edges than the estimates obtained through PLVAR and
PLVAR2. The difference between the likelihood-based methods and the PLVAR
methods is smaller but still evident. The MSEs on the test data were somewhat
equal for all the methods except for LS and LS2. The MSEs obtained through LS2
were slightly higher that the MSEs obtained through the likelihood-based and
the PLVAR methods. The MSEs produced by LS were about 100 times as high.

The computation times were minimal for LS and LS2 but this makes hardly a
difference if the goal is to learn a sparse VAR model. The computation time
required by PLVAR2 was an order of magnitude shorter than the time required by
LASSO2 and SCAD2. The PLVAR method, tasked with the additional challenge of
estimating the lag length, was slower than PLVAR2 but still faster than LASSO2
and SCAD2.

The example in Figure \ref{figure:results_matrix_visualization} highlights the
difference in sparsity of the model estimates obtained through PLVAR2 and the
likelihood-based methods. In this particular case, it is evident that PLVAR2
produced the sparsest precision matrix and lag matrices, corresponding to the
contemporaneous network and temporal network, respectively. One can find some
common structures in all the precision matrices but the structures in the
likelihood-based estimates are much denser.

%===============================================================================

\section{Discussion}
\label{section:discussion}

In this paper we have extended the idea of structure learning of Gaussian
graphical models via pseudo-likelihood and Bayes factors from the
cross-sectional domain into the time-series domain of VAR models. We have
presented a novel method, PLVAR, that is able to infer the complete VAR model
structure, including the lag length, in a single run. We have shown that the
PLVAR method produces a consistent structure estimator in the class of VAR
models. We have also presented an iterative maximum likelihood procedure for
inferring the model parameters for a given sparse VAR structure, providing a
complete VAR model learning pipeline.

The experimental results on synthetic data suggest that the PLVAR method is
better suited for recovering sparse VAR structures than the available
likelihood-based methods, a result that is in line with previous research
\cite{Leppa-aho-2017,Pensar-2017}. The results on EEG data provide further evidence on the feasibility of the
PLVAR method. Even though LASSO and SCAD reached quite similar MSEs, PLVAR was
able to do so while producing much sparser models and using significantly less
computation time. The conventional LS and LS2 methods were even faster but they
are unable to learn sparse VAR models. Also, the MSEs produced by LS and LS2
were the highest among all the methods under consideration. All in all, the results show that PLVAR is a highly competitive method
for high-dimensional VAR structure learning, from both an accuracy point of
view and a purely computational point of view.

It should be noted that the PLVAR implementation used in the experiments is
purely sequential. It would be relatively straightforward to make the search
algorithm run in parallel, thereby reducing the computation time up to
$1/n_\text{cpus}$ where $n_\text{cpus}$ is the number of CPU cores available on
the system. While it is common for modern desktop and laptop workstations to
be equipped with 4 to 8 physical cores, parallelization can be taken much
further by utilizing cloud computing services. Several cloud vendors offer
high-power virtual machines with tens of CPUs and the possibility of organizing
them into computing clusters. This makes PLVAR highly competitive against
methods that cannot be easily parallelized, especially when the number of
variables in the model is large.

In summary, PLVAR is an accurate and fast method for learning VAR models. It is
able to produce sparse estimates of the model structure and thereby allows the
model to be interpreted to gain additional knowledge in the problem domain.

In future work, we plan to apply the PLVAR method to domains that make
extensive use of VAR modeling and where improvements in estimation accuracy and
speed have the potential to enhance the end results significantly. An example
of these domains is brain research where brain connectivity is measured in
terms of Granger causality index, directed transfer function, and partial
directed coherence. Since these measures are derived from VAR models estimated
from the data, one can expect more accurate estimates to result in more
accurate measures of connectivity. Moreover, the enhanced possibilites in VAR
model fitting introduced here allow more complex models to be studied. Further
research with applications to other high-dimensional problems is also
encouraged.

%===============================================================================

\section*{Supplementary material}

The R code written for the experiments is publicly available in the repository
https://github.com/kimmo-suotsalo/plvar. The EEG data can be freely obtained
from PhysioNet \cite{Goldberger-2000} at
https://www.physionet.org/content/chbmit/1.0.0/.

%===============================================================================

\section*{Acknowledgements}
    This work was supported by RIKEN Special Postdoctoral Researcher Program
    (Yingying Xu). The authors wish to acknowledge CSC – IT Center for Science,
    Finland, for computational resources.

%===============================================================================

\section*{Appendices}

\section*{A Pseudo-code for the search algorithm}

The pseudo-code for the greedy search procedure is presented in Algorithm
\ref{algorithm:plvar}. The algorithm has been published previously in
\cite{Leppa-aho-2017} and was originally introduced in \cite{Pensar-2017},
based on \cite{Tsamardinos-2006}.

\begin{algorithm}
\caption{Greedy search for finding the temporal Markov blanket of node $i$}
\label{algorithm:plvar}
\begin{algorithmic}[1]
\Function{find\_mb}{$i, \mathbf{Z}, k$}
    \State mb($i$), $\widehat{\text{mb}}$($i$) $\leftarrow \emptyset$
    \While{$\widehat{\text{mb}}$($i$) has changed AND $|\text{mb}(i)| < N-1$ }
        \State C $\leftarrow V \setminus \text{fa}(i)$
        \State \Comment{above, $V$ contains lagged variables from $t-1$ to
        $t-k$}
        \State mb($i$) $\leftarrow \widehat{\text{mb}}(i)$
        \For{$j \in C$}
            \If{$p(\mathbf{Z}_i | \mathbf{Z}_{\text{mb}(i) \cup \{j\}},k)$ $>$
            $p(\mathbf{Z}_i |
            \mathbf{Z}_{\widehat{\text{mb}}(i)},k)$}
                \State $\widehat{\text{mb}}$($i$) $\leftarrow$ mb($i$) $\cup
                \{j\}$
            \EndIf
        \EndFor
        \While{$\widehat{\text{mb}}$($i$) has changed}
            \State mb($i$) $\leftarrow \widehat{\text{mb}}(i)$
            \For{$j \in \text{mb}(i)$}
                \If{$p(\mathbf{Z}_i | \mathbf{Z}_{\text{mb}(i) \setminus
                \{j\}},k)$ $>$ $p(\mathbf{Z}_i
                | \mathbf{Z}_{\widehat{\text{mb}}(i)},k)$}
                    \State $\widehat{\text{mb}}$($i$) $\leftarrow$ mb($i$)
                    $\setminus \{j\}$
                \EndIf
            \EndFor
        \EndWhile
    \EndWhile
    \State \Return $\widehat{\text{mb}}(i)$
\EndFunction
\end{algorithmic}
\end{algorithm}

\section*{B Consistency proof}
The below result states that the modified FMPL given in
\eqref{eqn:fmpl_temp_opt} is consistent for inferring the temporal Markov
blankets in the considered model class. The proof is adapted from
\cite{Leppa-aho-2017} where the authors use a similar reasoning for Gaussian
graphical models with cross-sectional data.
\begin{theorem} \label{lemma:consistency_of_local_pa_fmpl}
    Define $\mathbf{Z}$ and \normalfont $p(\mathbf{Z}_i | \mathbf{Z}_{\text{mb}(i)}, k)$
    \textit{as in Equation
         \eqref{eq:Zmatrix} and (\ref{eqn:local_fmpl_temp}), respectively. Let $S = \{ d+1,\ldots,(k+1)d \}$ contain the true temporal Markov blanket $\text{mb}(i)^*$ or, equivalently, let $k \geq k^*$ where $k^*$ is the true lag length. Then,} \normalfont
\begin{equation}\label{eq:tmb_res}
\plim_{N \rightarrow \infty} \argmax_{\text{mb}(i) \subseteq S} p(\mathbf{Z}_i | \mathbf{Z}_{\text{mb}(i)}, k) = \text{mb}(i)^*.
 \end{equation}
\end{theorem}
\begin{proof}
Under the current assumption of i.i.d. $N(\zerovec,\Sigma)$ error terms, it can
be shown that $\yvec_t$ is a Gaussian process, meaning that the subcollections
$\yvec_t, \yvec_{t-1},\ldots, \yvec_{t-k}$ have a multivariate normal
distribution \cite{Lutkepohl-2005}. Thus, $\mathbf{Z}$ can be considered a data
matrix generated from a multivariate normal distribution over the variables
$Z_1,\ldots , Z_{(k+1)d}$ where the dependence structure of the model is
determined by the GVAR structure. In particular, following from the directed
Markov property in (\ref{eqn:temporal_markov_blanket}), the dependence
structure will be such that $$Z_i
\perp Z_{\text{mb}(i)^*} \mid Z_{S \setminus \text{mb}(i)^*},$$
where $i \in \{ 1,\ldots , d\}$.
Following a similar line of reasoning as in \cite{Leppa-aho-2017}, we can now show through Lemma \ref{lemma:missing_nodes} and Lemma \ref{lemma:extra_nodes} that \eqref{eq:tmb_res} holds. More specifically, Lemma \ref{lemma:missing_nodes} guarantees that
    as $N \rightarrow \infty$, $$p(\mathbf{Z}_i | \mathbf{Z}_{\text{mb}(i)}, k)
    < p(\mathbf{Z}_i |
    \mathbf{Z}_{\text{mb}(i)^*}, k)$$ for any $\text{mb}(i)$ that does not
    contain all
    the nodes in $\text{mb}(i)^*$. Similarly, Lemma \ref{lemma:extra_nodes}
    guarantees that  as $N \rightarrow \infty$, $$p(\mathbf{Z}_i
    |\mathbf{Z}_{\text{mb}(i)}, k)
    < p(\mathbf{Z}_i | \mathbf{Z}_{\text{mb}(i)^*}, k)$$ for any $\text{mb}(i)$
    that contains
    one or more nodes not in $\text{mb}(i)^*$.
\end{proof}
The proofs of Lemma \ref{lemma:missing_nodes} and Lemma \ref{lemma:extra_nodes}
follow closely those of \cite{Leppa-aho-2017}. However, for completeness, we
give here the proofs in detail.

\begin{lemma} \label{lemma:missing_nodes}
    Let \normalfont $\text{mb}(i) \subset \text{mb}(i)^*$ \textit{and}
    $A \subset V \setminus \text{fa}(i)^*$. \textit{Then}
    \begin{equation}
    \plim_{N \rightarrow \infty} \log \frac{p(\mathbf{Z}_i |
    \mathbf{Z}_{\text{mb}(i)^* \cup A},
        k)}{p(\mathbf{Z}_i | \mathbf{Z}_{\text{mb}(i) \cup A}, k)} = \infty.
    \end{equation}
\end{lemma}

\begin{proof}
    By defining $p_i \triangleq |\text{mb}(i)^* \cup A|$ and $a \triangleq
    |\text{mb}(i) \cup A| - p_i$, we have that
    \begin{align} \label{eqn:missing_nodes}
        &\log \frac{p(\mathbf{Z}_i | \mathbf{Z}_{\text{mb}(i)^* \cup A},
            k)}{p(\mathbf{Z}_i | \mathbf{Z}_{\text{mb}(i) \cup A}, k)}
        =
        \log \frac{ \Gamma\left(\frac{N-k+p_i}{2}\right) }
        { \Gamma\left(\frac{N-k+p_i+a}{2}\right) }
        + \nonumber \\
        &\log \frac{ \Gamma\left(\frac{p_i+a+1}{2}\right) }
        { \Gamma\left(\frac{p_i+1}{2}\right) }
        +
        a \log(N-k)
        - \nonumber \\
        &\frac{N-k-1}{2} \log
        \frac{ \left|\mathbf{S_\text{fa($i$)$^* \cup A$}}\right|
            \left|\mathbf{S_\text{mb($i$)$ \cup A$}}\right|}
        { \left|\mathbf{S_\text{mb($i$)$^* \cup A$}}\right|
            \left|\mathbf{S_\text{fa($i$)$ \cup A$}}\right|}.
    \end{align}
    The first term on the right-hand side equals
    \begin{equation}
        \log \Gamma (M) - \log \Gamma \left(M + \frac{a}{2}\right)
    \end{equation}
    where $M = (N - k + p_i) / 2$. As $N \rightarrow \infty$, $M \rightarrow
    \infty$ and by Stirlings's formula
    \begin{align}
        &\log \Gamma (M) - \log \Gamma \left(M + \frac{a}{2}\right)
        =
        \left(M - \frac{1}{2}\right) \log M - M \nonumber \\
        &-\left[ \left(M + \frac{a}{2} - \frac{1}{2} \right) \log \left(M +
        \frac{a}{2}\right) - \left(M + \frac{a}{2}\right) \right] +
        O(1).
    \end{align}
    Here, the right-hand side can be written as
    \begin{align}
    M \log \frac{M}{M + \frac{a}{2}}
    +
    \frac{1}{2} \log \frac{M + \frac{a}{2}}{M}
    -
    \frac{a}{2} \log \left(M + \frac{a}{2}\right)
    +
    O(1)
    \end{align}
    where we have included the term $a / 2$ in $O(1)$. Since
    \begin{align}
    &M \log \frac{M}{M + \frac{a}{2}}
    =
    \frac{1}{2} \log \left(\frac{1}{1 + \frac{a}{2M}}\right)^{2M}
    \xrightarrow[M \rightarrow \infty] \nonumber \\
    &\frac{1}{2} \log e^{-a}
    =
    - \frac{a}{2} = O(1)
    \end{align}
    and
    \begin{align}
    \frac{1}{2} \log \frac{M + \frac{a}{2}}{M}
    =
    \frac{1}{2} \log \left(1 + \frac{a}{2M}\right) \xrightarrow[M \rightarrow
    \infty]{} 0,
    \end{align}
    we have that
    \begin{align}
    \log \frac{ \Gamma\left(\frac{N-k+p_i}{2}\right) }
    { \Gamma\left(\frac{N-k+p_i+a}{2}\right) }
    \xrightarrow[N \rightarrow \infty]{}
    -\frac{a}{2} \log \left(\frac{N-k+p_i+a}{2} \right) + O(1).
    \end{align}

    The second term on the right-hand side of Equation
    (\ref{eqn:missing_nodes}) can be included in $O(1)$ since it is constant
    with respect to $N$. The fourth term
    \begin{align}
    - \frac{N-k-1}{2} \log
    \frac{ \left|\mathbf{S_\text{fa($i$)$^* \cup A$}}\right|
        \left|\mathbf{S_\text{mb($i$)$ \cup A$}}\right|}
    { \left|\mathbf{S_\text{mb($i$)$^* \cup A$}}\right|
        \left|\mathbf{S_\text{fa($i$)$ \cup A$}}\right|}
    \end{align}
    can be expressed using the covariance $\boldsymbol{\Sigma}$: since
    $\mathbf{S} = \mathbf{Z}^\top \mathbf{Z} = N
    \widehat{\boldsymbol{\Sigma}}_\text{ML}$, it
    holds that
    \begin{align}
    \frac{ \left|\mathbf{S_\text{fa($i$)$^* \cup A$}}\right|
        \left|\mathbf{S_\text{mb($i$)$ \cup A$}}\right|}
    { \left|\mathbf{S_\text{mb($i$)$^* \cup A$}}\right|
        \left|\mathbf{S_\text{fa($i$)$ \cup A$}}\right|}
    \xrightarrow[N \rightarrow \infty]{}
    \frac{ \left|\boldsymbol{\Sigma_\text{fa($i$)$^* \cup A$}}\right|
        \left|\boldsymbol{\Sigma_\text{mb($i$)$ \cup A$}}\right|}
    { \left|\boldsymbol{\Sigma_\text{mb($i$)$^* \cup A$}}\right|
    \left|\boldsymbol{\Sigma_\text{fa($i$)$ \cup A$}}\right|}
    \end{align}
    in probability. Let us now partition $\boldsymbol{\Sigma_\text{fa($i$)$^*
    \cup A$}}$ as
    \begin{align}
    \boldsymbol{\Sigma_\text{fa($i$)$^* \cup A$}}
     =
     \begin{bmatrix}
        \sigma^2(z_i) & \sigma(z_i, \zvec_\text{mb($i$)$^* \cup A$}) \\
        \sigma(z_i, \zvec_\text{mb($i$)$^* \cup A$})^\top &
        \boldsymbol{\Sigma}_\text{mb($i$)$^* \cup A$}
     \end{bmatrix}
    \end{align}
    where $\sigma(z_i, \zvec_\text{mb($i$)$^* \cup A$})$ denotes a row vector
    whose elements are the covariances between the variable $z_i$ and each of
    the variables in $\text{mb($i$)$^* \cup A$}$. The determinant of this
    partitioned matrix is
    \begin{align}
    &\left|\boldsymbol{\Sigma}_\text{mb($i$)$^* \cup A$}\right| \cdot
    \nonumber
    \\
    &\left[ \sigma^2(z_i) - \sigma(z_i, \zvec_\text{mb($i$)$^* \cup A$}) \left(
    \boldsymbol{\Sigma}_\text{mb($i$)$^* \cup A$}
    \right)^{-1} \sigma(z_i, \zvec_\text{mb($i$)$^* \cup A$})^\top \right]
    \end{align}
    where the term
    \begin{align}
    &\sigma(z_i, \zvec_\text{mb($i$)$^* \cup A$}) \left(
    \boldsymbol{\Sigma}_\text{mb($i$)$^* \cup A$}
    \right)^{-1} \sigma(z_i, \zvec_\text{mb($i$)$^* \cup A$})^\top \triangleq
    \nonumber \\
    &\sigma(\hat{z}_i [\zvec_\text{mb($i$)$^* \cup A$}])
    \end{align}
    is the variance of the linear LS predictor of $z_i$ from
    $\zvec_\text{mb($i$)$^* \cup A$}$.
    By definition, the residual variance of $z_i$ after subtracting the
    variance of the linear LS predictor is the partial variance $\sigma(z_i |
    \zvec_\text{mb($i$)$^* \cup A$})$; hence,
    \begin{align}
    \left|\boldsymbol{\Sigma}_\text{fa($i$)$^* \cup A$}\right|
    =
    \left|\boldsymbol{\Sigma}_\text{mb($i$)$^* \cup A$}\right|
    \sigma(z_i | \zvec_\text{mb($i$)$^* \cup A$})
    \end{align}
    which implies that
    \begin{align}
    \frac{ \left|\boldsymbol{\Sigma_\text{fa($i$)$^* \cup A$}}\right|
        \left|\boldsymbol{\Sigma_\text{mb($i$)$ \cup A$}}\right|}
    { \left|\boldsymbol{\Sigma_\text{mb($i$)$^* \cup A$}}\right|
        \left|\boldsymbol{\Sigma_\text{fa($i$)$ \cup A$}}\right|}
    =
    \frac{\sigma(z_i | \zvec_\text{mb($i$)$^* \cup A$})}{\sigma(z_i |
    \zvec_\text{mb($i$)$ \cup A$})}.
    \end{align}

    Combining the results so far, Equation (\ref{eqn:missing_nodes}) can
    be rewritten as
    \begin{align}
    &\log \frac{p(\mathbf{Z}_i | \mathbf{Z}_{\text{mb}(i)^* \cup A},
        k)}{p(\mathbf{Z}_i | \mathbf{Z}_{\text{mb}(i) \cup A}, k)}
    =
    -\frac{a}{2} \log \left(\frac{N-k+p_i+a}{2} \right) + \nonumber \\
    &a \log(N-k) - \frac{N-k-1}{2} \log \frac{\sigma(z_i |
    \zvec_\text{mb($i$)$^*
    \cup A$})}{\sigma(z_i | \zvec_\text{mb($i$)$ \cup A$})} + O(1).
    \end{align}
    The second term on the right hand side equals
    \begin{align}
    & \frac{a}{2} \log (N-k)^2 = \frac{a}{2} \log (N-k) + \frac{a}{2} \log (N-k)
    \end{align}
    and therefore
    \begin{align}
    &\log \frac{p(\mathbf{Z}_i | \mathbf{Z}_{\text{mb}(i)^* \cup A},
        k)}{p(\mathbf{Z}_i | \mathbf{Z}_{\text{mb}(i) \cup A}, k)}
    = \nonumber \\
    &\frac{a}{2} \log(N-k) -
     \frac{N}{2} \log \frac{\sigma(z_i |
    \zvec_\text{mb($i$)$^* \cup A$})}{\sigma(z_i | \zvec_\text{mb($i$)$ \cup
    A$})} + O(1).
    \end{align}
    Now, consider the argument of the last logarithm term. Assume first that
    $\text{mb($i$)} \neq \emptyset$ and let $B \triangleq \text{mb($i$)}^*
    \setminus \text{mb($i$)}$. Then
    \begin{align}
    &\sigma(\hat{z}_i [\zvec_\text{mb($i$)$^* \cup A$}])
    =
    \sigma(\hat{z}_i [\zvec_\text{mb($i$)$ \cup A \cup B$}])
    = \nonumber \\
    &\sigma(\hat{z}_i [\zvec_\text{mb($i$)$ \cup A$}])
    +
    \sigma(\hat{z}_i [\zvec_B - \hat{\zvec}_B[\zvec_\text{mb($i$)$ \cup A$}]])
    \end{align}
    where the last term $\sigma(\hat{z}_i [\zvec_B -
    \hat{\zvec}_B[\zvec_\text{mb($i$)$ \cup A$}]])
    > 0$
    since $B \subset\text{mb($i$)}^*$. Hence,
    \begin{align} \label{ineqn:variances}
    &\sigma(\hat{z}_i [\zvec_\text{mb($i$)$^* \cup A$}])
    >
    \sigma(\hat{z}_i [\zvec_\text{mb($i$)$ \cup A$}])
    \end{align}
    which is equivalent to
    \begin{align}
    &\sigma^2(z_i) - \sigma(\hat{z}_i [\zvec_\text{mb($i$)$^* \cup A$}])
    <
    \sigma^2(z_i) -\sigma(\hat{z}_i [\zvec_\text{mb($i$)$ \cup A$}]).
    \end{align}
    This implies that
    \begin{align}
    \sigma(z_i | \zvec_\text{mb($i$)$^* \cup A$})
    <
    \sigma(z_i | \zvec_\text{mb($i$)$ \cup A$})
    \end{align}
    and therefore the last logarithm term
    \begin{align}
    &- \frac{N}{2} \log \frac{\sigma(z_i |
    \zvec_\text{mb($i$)$^* \cup A$})}{\sigma(z_i | \zvec_\text{mb($i$)$
    \cup A$})}
    \xrightarrow[N \rightarrow \infty]{} \infty.
    \end{align}
    Assume then that $\text{mb($i$)} = \emptyset$. This implies that
    \begin{align}
    \sigma(\hat{z}_i [\zvec_A])
    +
    \sigma(\hat{z}_i [\zvec_\text{mb($i$)$^*$} -
    \hat{z}_\text{mb($i$)$^*$}[\zvec_A]])
    >
    \sigma(\hat{z}_i [\zvec_A])
    \end{align}
    which is equivalent to
    \begin{align}
    &\sigma(\hat{z}_i [\zvec_\text{mb($i$)$^* \cup A$}])
    >
    \sigma(\hat{z}_i [\zvec_A])
    \end{align}
    and further to Inequation (\ref{ineqn:variances}). Hence, the last
    logarithm term behaves exactly as for $\text{mb($i$)} \neq \emptyset$.

    Since $\text{mb}(i) \subset \text{mb}(i)^*$, $a < 0$, and the second-last
    logarithm term
    \begin{align}
    \frac{a}{2} \log(N-k) \xrightarrow[N \rightarrow \infty]{} -\infty.
    \end{align}
    Combining all the results, we have that
    \begin{align}
    &\log \frac{p(\mathbf{Z}_i | \mathbf{Z}_{\text{mb}(i)^* \cup A},
        k)}{p(\mathbf{Z}_i | \mathbf{Z}_{\text{mb}(i) \cup A}, k)}
    = \nonumber \\
    &\frac{a}{2} \log(N-k) -
    \frac{N}{2} \log \frac{\sigma(z_i |
    \zvec_\text{mb($i$)$^* \cup A$})}{\sigma(z_i | \zvec_\text{mb($i$)$ \cup
        A$})} + O(1)
    \end{align}
    where the first and the second term on the right-hand side tend to
    $-\infty$ and $\infty$, respectively. However, since $-\log N$ decreases
    slower than $N$ increases, the right-hand side of the equation tends to
    $\infty$. This proves the statement of Lemma \ref{lemma:missing_nodes}.
    \end{proof}

\begin{lemma} \label{lemma:extra_nodes}
    Let \normalfont $\text{mb}(i)^* \subset \text{mb}(i)$. \textit{Then}
    \begin{equation}
    \plim_{N \rightarrow \infty} \log \frac{p(\mathbf{Z}_i |
    \mathbf{Z}_{\text{mb}(i)^*},
        k)}{p(\mathbf{Z}_i | \mathbf{Z}_{\text{mb}(i)}, k)} = \infty.
    \end{equation}
\end{lemma}

\begin{proof}

    By defining $p_i \triangleq |\text{mb}(i)^*|$ and $a \triangleq
    |\text{mb}(i)| - p_i$, we have that
    \begin{align} \label{eqn:extra_nodes}
    &\log \frac{p(\mathbf{Z}_i | \mathbf{Z}_{\text{mb}(i)^*},
        k)}{p(\mathbf{Z}_i | \mathbf{Z}_{\text{mb}(i)}, k)}
    =
    \log \frac{ \Gamma\left(\frac{N-k+p_i}{2}\right) }
    { \Gamma\left(\frac{N-k+p_i+a}{2}\right) }
    + \nonumber \\
    &\log \frac{ \Gamma\left(\frac{p_i+a+1}{2}\right) }
    { \Gamma\left(\frac{p_i+1}{2}\right) }
    +
    a \log(N-k)
    - \nonumber \\
    &\frac{N-k-1}{2} \log
    \frac{ \left|\mathbf{S_\text{fa($i$)$^*$}}\right|
        \left|\mathbf{S_\text{mb($i$)$$}}\right|}
    { \left|\mathbf{S_\text{mb($i$)$^*$}}\right|
        \left|\mathbf{S_\text{fa($i$)$$}}\right|}.
    \end{align}
    This equation is identical to (\ref{eqn:missing_nodes}) except for the
    last term on the right-hand side. As $N \rightarrow \infty$, this term
    tends to
    \begin{align}
    &-N \log
    \frac{ \left|\mathbf{S_\text{fa($i$)$^*$}}\right|
        \left|\mathbf{S_\text{mb($i$)$$}}\right|}
    { \left|\mathbf{S_\text{mb($i$)$^*$}}\right|
        \left|\mathbf{S_\text{fa($i$)$$}}\right|}
    \cdot \frac{1}{2} + O(1).
    \end{align}
    Now select $B \subset \text{mb($i$)}$ so that $\text{mb($i$)} =
    \text{mb($i$)}^* \cup B$. Then
    \begin{align}
    &-N \log
    \frac{ \left|\mathbf{S_\text{fa($i$)$^*$}}\right|
        \left|\mathbf{S_\text{mb($i$)$$}}\right|}
    { \left|\mathbf{S_\text{mb($i$)$^*$}}\right|
        \left|\mathbf{S_\text{fa($i$)$$}}\right|}
    =
    \text{dev}(z_i \perp \!\!\!\!\!\! \perp \zvec_B \: | \:
    \zvec_{\text{mb($i$)}^*} )
    \end{align}
    where, asymptotically, the deviance $\text{dev}(z_i \perp \!\!\!\!\!\!
    \perp \zvec_B \: |
    \:
    \zvec_{\text{mb($i$)}^*} )$ follows \cite{Whittaker-1990} a
    $\chi^2$ distribution. This implies that the deviance is
    bounded in probability. Consequently,
    \begin{align}
    - \frac{N-k-1}{2} \log
    \frac{ \left|\mathbf{S_\text{fa($i$)$^*$}}\right|
        \left|\mathbf{S_\text{mb($i$)$$}}\right|}
    { \left|\mathbf{S_\text{mb($i$)$^*$}}\right|
        \left|\mathbf{S_\text{fa($i$)$$}}\right|}
    = O(1).
    \end{align}

    Now, applying the same reasoning as for Lemma \ref{lemma:missing_nodes}, we
    arrive at
    \begin{align}
    &\log \frac{p(\mathbf{Z}_i | \mathbf{Z}_{\text{mb}(i)^*},
        k)}{p(\mathbf{Z}_i | \mathbf{Z}_{\text{mb}(i)}, k)}
    =
    \frac{a}{2} \log(N-k) + O(1).
    \end{align}
    Since $\text{mb}(i)^* \subset \text{mb}(i)$, $a > 0$, and the right-hand
    side of the equation tends to $\infty$. This proves the statement of Lemma
    \ref{lemma:extra_nodes}.

\end{proof}

\bibliographystyle{ieeetr}
\bibliography{refs}

\end{document}